\pdfoutput=1
\RequirePackage{ifpdf}
\ifpdf 
\documentclass[pdftex]{sigma}
\else
\documentclass{sigma}
\fi

\usepackage{enumerate} 
\usepackage{upgreek} 
\usepackage{tikz} 
\usetikzlibrary{matrix,graphs,arrows,positioning,calc,decorations.markings,shapes.symbols}
\usepackage{tikz-cd}

\numberwithin{equation}{section}

\newtheorem{Theorem}{Theorem}[section]
\newtheorem{Lemma}[Theorem]{Lemma}
\newtheorem{Proposition}[Theorem]{Proposition}
{ \theoremstyle{definition}
\newtheorem{Remark}[Theorem]{Remark} }

\begin{document}

\allowdisplaybreaks

\newcommand{\arXivNumber}{1804.10341}

\renewcommand{\thefootnote}{}

\renewcommand{\PaperNumber}{075}

\FirstPageHeading

\ShortArticleName{On Some Applications of Sakai's Geometric Theory of Discrete Painlev\'e Equations}

\ArticleName{On Some Applications of Sakai's Geometric Theory\\ of Discrete Painlev\'e Equations\footnote{This paper is a~contribution to the Special Issue on Painlev\'e Equations and Applications in Memory of Andrei Kapaev. The full collection is available at \href{https://www.emis.de/journals/SIGMA/Kapaev.html}{https://www.emis.de/journals/SIGMA/Kapaev.html}}}

\Author{Anton DZHAMAY~$^\dag$ and Tomoyuki TAKENAWA~$^\ddag$}

\AuthorNameForHeading{A.~Dzhamay and T.~Takenawa}

\Address{$^\dag$~School of Mathematical Sciences, The University of Northern Colorado,\\
\hphantom{$^\dag$}~Campus Box 122, 501 20th Street, Greeley, CO 80639, USA}
\EmailD{\href{mailto:adzham@unco.edu}{adzham@unco.edu}}

\Address{$^\ddag$~Faculty of Marine Technology, Tokyo University of Marine Science and Technology,\\
\hphantom{$^\ddag$}~2-1-6 Etchujima, Koto-ku Tokyo, 135-8533, Japan}
\EmailD{\href{mailto:takenawa@kaiyodai.ac.jp}{takenawa@kaiyodai.ac.jp}}

\ArticleDates{Received April 30, 2018, in final form July 14, 2018; Published online July 21, 2018}

\Abstract{Although the theory of discrete Painlev\'e (dP) equations is rather young, more and more examples of such equations appear in interesting and important applications. Thus, it is essential to be able to recognize these equations, to be able to identify their type, and to see where they belong in the classification scheme. The definite classification scheme for dP equations was proposed by H.~Sakai, who used geometric ideas to identify 22 different classes of these equations. However, in a major contrast with the theory of ordinary differential Painlev\'e equations, there are infinitely many non-equivalent discrete equations in each class. Thus, there is no general form for a dP equation in each class, although some nice canonical examples in each equation class are known. The main objective of this paper is to illustrate that, in addition to providing the classification scheme, the geometric ideas of Sakai give us a powerful tool to study dP equations. We consider a~very complicated example of a dP equation that describes a simple Schlesinger transformation of a Fuchsian system and we show how this equation can be identified with a much simpler canonical example of the dP equation of the same type and moreover, we give an explicit change of coordinates transforming one equation into the other. Among our main tools are the birational representation of the affine Weyl symmetry group of the equation and the period map. Even though we focus on a concrete example, the techniques that we use are general and can be easily adapted to other examples.}


\Classification{34M55; 34M56; 14E07}

\renewcommand{\thefootnote}{\arabic{footnote}}
\setcounter{footnote}{0}

\begin{flushright}\it
To the memory of A.~Kapaev
\end{flushright}

\section{Introduction} \label{sec:introduction}

Recall that the original motivation behind P.~Painlev\'e's study \cite{Pai:1902:SLEDDSOEDSDLGEU} of the class of equations that are now know as \emph{differential Painlev\'e equations} was to try to define new, purely nonlinear, special functions as general solutions of such equations, see also \cite{IwaKimShiYos:1991:FGP}. And indeed, these solutions, now known as the \emph{Painlev\'e transcendents}, are playing an increasingly important role in applications. Probably the most important example of this is the famous \emph{Tracy--Widom} distribution~\cite{TraWid:1994:LDATAK} from random matrix theory that can be expressed in terms of the Hastings--McLeod solution of the Painlev\'e-II equation.

The theory of discrete Painlev\'e equations is much more recent. Even though many examples of what is now known as discrete Painlev\'e equations appeared in various branches of mathema\-tics starting as early as the 1930s, the systematic study of discrete Painlev\'e equations started in the early 1990s in the work of B.~Grammaticos and A.~Ramani, who first, together with V.~Papageorgiou, introduced the notion of a~\emph{singularity confinement} \cite{GraRamPap:1991:IMHPP} as a discrete analogue of the Painlev\'e property and then, together with J.~Hietarinta~\cite{RamGraHie:1991:DVOTPE} applied it to systematically create examples of discrete Painlev\'e equations as deautonomizations of the so-called Quispel--Roberts--Thompson (QRT) mappings~\cite{QuiRobTho:1988:IMASE}, see the review~\cite{GraRam:2004:DPER}.

The next important step in the theory of discrete Painlev\'e equations was done by H.~Sakai who, in the fundamental paper~\cite{Sak:2001:RSAWARSGPE}, extended to the discrete case the geometric approach of K.~Okamoto \cite{Oka:1979:SLFAAEDSOAPCFDPP} to differential Painlev\'e equations. In particular, Sakai gave the definite classification scheme of discrete Painlev\'e equations based on the algebro-geometric classification of generalized Halphen surfaces $\mathcal{X}$ of index $0$ on which these equations are naturally regularized. In this classification, to each equation corresponds two ``dual'' affine Dynkin diagrams $(\mathcal{D}_{1},\mathcal{D}_{2})$. The first diagram $\mathcal{D}_{1}$ describes the \emph{geometry} of the equation by encoding the configuration of points that we blowup to obtain the surface $\mathcal{X}$, known as the \emph{Okamoto space of initial conditions} of the equation; the corresponding Sakai's classification scheme is shown on Fig.~\ref{fig:classification-surface}. The second diagram $\mathcal{D}_{2}$ describes the \emph{symmetry} structure of the equation in terms of the extended affine Weyl group $\widetilde{W}(\mathcal{D}_{2})$ associated to this diagram; the corresponding Sakai's classification scheme is shown on Fig.~\ref{fig:classification-symmetry}.

\begin{Remark} The degeneration arrows on Figs.~\ref{fig:classification-surface} and \ref{fig:classification-symmetry} are taken from E.~Rains paper~\cite{Rai:2013:GHSORS} and are corrections, for the characteristic zero, to the original Sakai's degeneration scheme given in~\cite{Sak:2001:RSAWARSGPE}. We are thankful to B.~Grammaticos, A.~Ramani, and R.~Willox for pointing this out to us, and also remark that the correct degeneration scheme for the symmetry structure (as in Fig.~\ref{fig:classification-symmetry}), with the exception of the~$A_{0}^{(1)}$-cases where no translations exist, has been used by B.~Grammaticos, A.~Ramani, and their collaborators since 2003, see~\cite{TamRamGra:2003:DI}.
\end{Remark}

Note that the connection between birational representation of affine Weyl groups and discrete dynamical systems of Painlev\'e type was studied in detail by M.~Noumi and Y.~Yamada, see, e.g.,~\cite{NouYam:1998:AWGDDSAPE}. In this setup, the equation itself is encoded by a translation element $\mathfrak{t}\in \widetilde{W}(\mathcal{D}_{2})$. Given that there are infinitely many non-equivalent translation elements, we do not have a~definite form of a discrete Painlev\'e equation of the given type, although many standard examples are well-known. A~comprehensive survey of the geometric aspects of discrete Painlev\'e equations is given in a very important recent paper~\cite{KajNouYam:2017:GAOPE}, see also references therein.

\begin{figure}[t]
 \centering\small
 \begin{tikzpicture}[>=stealth,scale=0.8]
 \node (e8e) at (0,4) {$A_{0}^{(1)}$};
 \node (a1qa) at (16,4) {$A_{7}^{(1)}$};
 \node (e8q) at (2,2) {$A_{0}^{(1)*}$};
 \node (e7q) at (4,2) {$A_{1}^{(1)}$};
 \node (e6q) at (6,2) {$A_{2}^{(1)}$};
 \node (d5q) at (8,2) {$A_{3}^{(1)}$};
 \node (a4q) at (10,2) {$A_{4}^{(1)}$};
 \node (a21q) at (12,2) {$A_{5}^{(1)}$};
 \node (a11q) at (14,2) {$A_{6}^{(1)}$};
 \node (a1q) at (16,2) {$A_{7}^{(1)}$};
 \node (a0q) at (18,2) {$A_{8}^{(1)}$};
 \node (e8d) at (4,0) {$A_{0}^{(1)**}$};
 \node (e7d) at (6,0) {$A_{1}^{(1)*}$};
 \node (e6d) at (8,0) {$A_{2}^{(1)*}$};
 \node (d4d) at (10,0) {$D_{4}^{(1)}$};
 \node (a3d) at (12,0) {$D_{5}^{(1)}$};
 \node (a11d) at (14,0) {$D_{6}^{(1)}$};
 \node (a1d) at (16,0) {$D_{7}^{(1)}$};
 \node (a0d) at (18,0) {$D_{8}^{(1)}$};
 \node (a2d) at (14,-2) {$E_{6}^{(1)}$};
 \node (a1da) at (16,-2) {$E_{7}^{(1)}$};
 \node (a0da) at (18,-2) {$E_{8}^{(1)}$};
 \draw[->] (e8e) -> (e8q); \draw[->] (a1qa) -> (a0d);
 \draw[->] (e8q) -> (e7q); \draw[->] (e8q) -> (e8d);
 \draw[->] (e7q) -> (e6q); \draw[->] (e7q) -> (e7d);
 \draw[->] (e6q) -> (d5q); \draw[->] (e6q) -> (e6d);
 \draw[->] (d5q) -> (a4q); \draw[->] (d5q) -> (d4d);
 \draw[->] (a4q) -> (a21q); \draw[->] (a4q) -> (a3d);
 \draw[->] (a21q) -> (a11q); \draw[->] (a21q) -> (a11d); \draw[->] (a21q) -> (a2d);
 \draw[->] (a11q) -> (a1q); \draw[->] (a11q) -> (a1d); \draw[->] (a11q) -> (a1qa); \draw[->] (a11q) -> (a1da);
 \draw[->] (a1q) -> (a0q); \draw[->] (a1q) -> (a0d); \draw[->] (a1q) -> (a0da);
 \draw[->] (e8d) -> (e7d);
 \draw[->] (e7d) -> (e6d);
 \draw[->] (e6d) -> (d4d);
 \draw[->] (d4d) -> (a3d);
 \draw[->] (a3d) -> (a11d); \draw[->] (a3d) -> (a2d);
 \draw[->] (a11d) -> (a1d); \draw[->] (a11d) -> (a1da);
 \draw[->] (a1d) -> (a0d); \draw[->] (a1d) -> (a0da);
 \draw[->] (a2d) -> (a1da); \draw[->] (a1da) -> (a0da);
 \end{tikzpicture}\vspace{-2mm}
\caption{Sakai's classification scheme for discrete Painlev\'e equations: the surface type.}\label{fig:classification-surface}
 \end{figure}
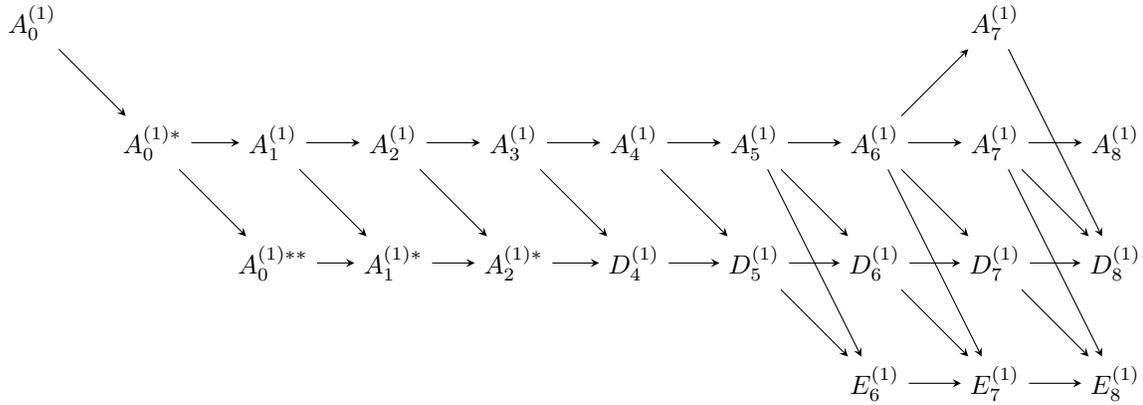

 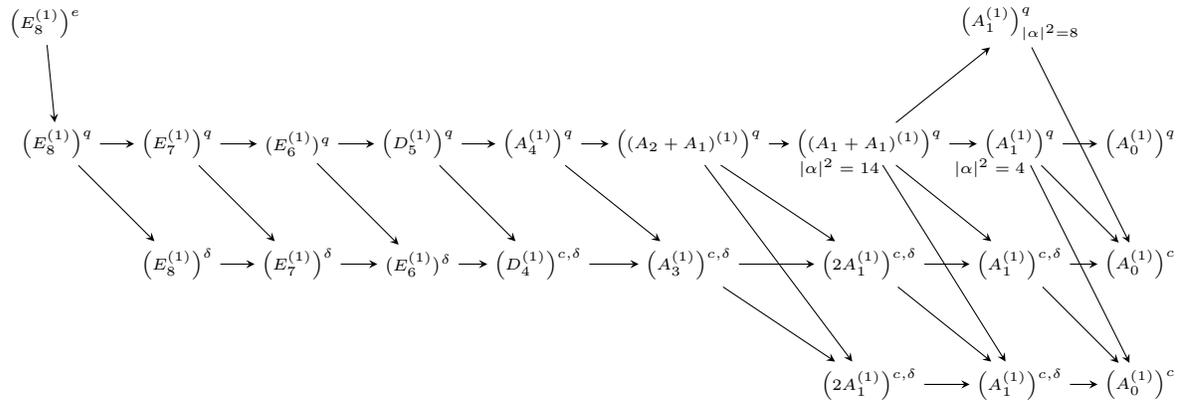
\begin{figure}[ht]
 \centering\tiny
 \begin{tikzpicture}[>=stealth,scale=0.8]
 \node (e8e) at (0.8,4) {$\left(E_{8}^{(1)}\right)^{e}$};
 \node (a1qa) at (17,4) {$\left(A_{1}^{(1)}\right)^{q}_{|\alpha|^{2}=8}$};
 \node (e8q) at (1,2) {$\left(E_{8}^{(1)}\right)^{q}$};
 \node (e7q) at (3,2) {$\left(E_{7}^{(1)}\right)^{q}$};
 \node (e6q) at (5,2) {$\big(E_{6}^{(1)}\big)^{q}$};
 \node (d5q) at (7,2) {$\left(D_{5}^{(1)}\right)^{q}$};
 \node (a4q) at (9,2) {$\left(A_{4}^{(1)}\right)^{q}$};
 \node (a21q) at (11.5,2) {$\left((A_{2} + A_{1})^{(1)}\right)^{q}$};
 \node (a11q) at (14.5,2) {$\left((A_{1} + A_{1})^{(1)}\right)^{q}$};
 \node at (14,1.6) {$|\alpha|^{2}=14$};
 \node (a1q) at (17,2) {$\left(A_{1}^{(1)}\right)^{q}$};
 \node at (16.5,1.6) {$|\alpha|^{2}=4$};
 \node (a0q) at (19,2) {$\left(A_{0}^{(1)}\right)^{q}$};
 \node (e8d) at (3,0) {$\left(E_{8}^{(1)}\right)^{\delta}$};
 \node (e7d) at (5,0) {$\left(E_{7}^{(1)}\right)^{\delta}$};
 \node (e6d) at (7,0) {$\big(E_{6}^{(1)}\big)^{\delta}$};
 \node (d4d) at (9,0) {$\left(D_{4}^{(1)}\right)^{c,\delta}$};
 \node (a3d) at (11.5,0) {$\left(A_{3}^{(1)}\right)^{c,\delta}$};
 \node (a11d) at (14.5,0) {$\left(2A_{1}^{(1)}\right)^{c,\delta}$};
 \node (a1d) at (17,0) {$\left(A_{1}^{(1)}\right)^{c,\delta}$};
 \node (a0d) at (19,0) {$\left(A_{0}^{(1)}\right)^{c}$};
 \node (a2d) at (14.5,-2) {$\left(2A_{1}^{(1)}\right)^{c,\delta}$};
 \node (a1da) at (17,-2) {$\left(A_{1}^{(1)}\right)^{c,\delta}$};
 \node (a0da) at (19,-2) {$\left(A_{0}^{(1)}\right)^{c}$};
 \draw[->] (e8e) -> (e8q); \draw[->] (a1qa) -> (a0d);
 \draw[->] (e8q) -> (e7q); \draw[->] (e8q) -> (e8d);
 \draw[->] (e7q) -> (e6q); \draw[->] (e7q) -> (e7d);
 \draw[->] (e6q) -> (d5q); \draw[->] (e6q) -> (e6d);
 \draw[->] (d5q) -> (a4q); \draw[->] (d5q) -> (d4d);
 \draw[->] (a4q) -> (a21q); \draw[->] (a4q) -> (a3d);
 \draw[->] (a21q) -> (a11q); \draw[->] (a21q) -> (a11d); \draw[->] (a21q) -> (a2d);
 \draw[->] (a11q) -> (a1q); \draw[->] (a11q) -> (a1d); \draw[->] (a11q) -> (a1qa); \draw[->] (a11q) -> (a1da);
 \draw[->] (a1q) -> (a0q); \draw[->] (a1q) -> (a0d); \draw[->] (a1q) -> (a0da);
 \draw[->] (e8d) -> (e7d);
 \draw[->] (e7d) -> (e6d);
 \draw[->] (e6d) -> (d4d);
 \draw[->] (d4d) -> (a3d);
 \draw[->] (a3d) -> (a11d); \draw[->] (a3d) -> (a2d);
 \draw[->] (a11d) -> (a1d); \draw[->] (a11d) -> (a1da);
 \draw[->] (a1d) -> (a0d); \draw[->] (a1d) -> (a0da);
 \draw[->] (a2d) -> (a1da); \draw[->] (a1da) -> (a0da);
 \end{tikzpicture}\vspace{-2mm}
\caption{Sakai's classification scheme for discrete Painlev\'e equations: the symmetry type.} \label{fig:classification-symmetry}
 \end{figure}

Similarly to their differential counterpart, discrete Painlev\'e equations appear in a wide range of applied problems, such as discrete symmetries of Fuchsian systems (Schlesinger and Okamoto transformations), B\"acklund transformations of differential Painlev\'e equations~\cite{Nou:2004:PETS}, reductions of lattice equations, and, especially, in discrete problems of random matrix type \cite{Bor:2003:DPDPE,BorBoy:2003:DOTFPIDOPE}. Thus, it is important for researchers working in these and other areas to be able to recognize an occurence of a discrete Painlev\'e equation, to understand where it fits in the classification scheme, and also to see whether it is equivalent to one of previously studied examples whose properties and special solutions are known. It turns out that Sakai's geometric theory provides a powerful set of tools to answer such questions. Although these techniques are known to the experts in the discrete Painlev\'e theory, many researchers working with applied problems involving discrete Painlev\'e equations may not realize the strength of the geometric approach. Thus, one of our goals for the present survey paper is to illustrate the power of this approach, providing enough details to make this paper of immediate practical value to a wider audience of researchers working with applications. For that reason, we tried to make the paper essentially self-contained. Although we work with one concrete example, the techniques that we use can be adapted, in a rather straightforward way, to other situations as well.

The example that we consider was first obtained by the authors in \cite{DzhSakTak:2013:DSTTHFADPE} and it describes an elementary two-point Schlesinger transformation of a Fuchsian system of spectral type $111,111,111$. According to P.~Boalch~\cite{Boa:2009:QADPE}, the resulting equation should be a discrete Painlev\'e equation d-$P\big(A_{2}^{(1)*}\big)$ of surface type $A_{2}^{(1)*}$ whose symmetry group is the extended affine Weyl group $\widetilde{W}\big(E_{6}^{(1)}\big)$. However, understanding this equation directly from the isomonodromic setting is not very easy.

In general, given a birational mapping $\psi$ that we expect to be of discrete Painlev\'e type, the problem of classifying it and possibly matching it with a known mapping $\varphi$ can be broken down into the following steps.
\begin{enumerate}[(a)]\itemsep=0pt
\item The mapping is regularized by successfully resolving all of its indeterminate points using the blowup procedure.
\item The resulting algebraic surface may not be minimal. If this happens, we need to find a~relatively minimal surface by blowing down some unnecessary $-1$-curves. This problem was considered for the autonomous case by one of the authors, together with A.S.~Carstea, in~\cite{CarTak:2013:ANOMORSOFBDS} and for the non-autonomous case in a recent preprint by T.~Mase~\cite{Mas:2017:SOSOICFNMOTP}. Fortunately, this issue does not appear for the present example.
\item As a result of the first two steps we get the type of our mapping. However, the choice of root lattices and the geometric realization of the anti-canonical divisor may be different from the standard case. Thus, we need to find a change of basis of the Picard lattice that will map the surface and the symmetry root lattices to the standard ones, and then find a~birational mapping (i.e., a change of coordinates) that induces this change of basis. For the present example this had been done in~\cite{DzhTak:2015:GAORFSTTDPE}.
\item At this point we have matched the geometry of the problem with the geometry of some model example of a discrete Painlev\'e equation of the same type, and now we can compare the dynamics. This is done by comparing the corresponding translation elements $\psi_{*}$ and $\varphi_{*}$ in the affine Weyl symmetry group. If these elements are conjugated, $\psi_{*} = \sigma_{*} \circ \varphi_{*} \circ \sigma_{*}^{-1}$, then the resulting dynamics are equivalent and the underlying birational mapping $\sigma$ gives the explicit change of variables transforming one equation into the other. In the present paper we primarily focus on this final step.
\end{enumerate}

In our example, even after steps (a)--(c) that give us natural isomonodromic coordinates $(x,y)$, the resulting equation is still very complicated and has the form
\begin{gather}
 \bar{x} = \frac{(\alpha(x,y) - \beta(x,y))
 \big(\alpha(x,y) x \big(\theta_{1}^{1} - \theta_{1}^{2}\big) +
 \big(1 + \theta_{0}^{2}\big)\big(x \big(y - \theta_{1}^{2}\big) + y \big(\theta_{0}^{1} - \theta_{0}^{2}\big)\big)\big) }{
 (\alpha(x,y) - \beta(x,y)) \big(x \big(y - \theta_{1}^{2}\big) + \big(\theta_{0}^{1} - \theta_{0}^{2}\big)y\big) -
 \alpha(x,y) \big(\theta_{1}^{1} + 1\big)\big(\theta_{0}^{1} - \theta_{0}^{2}\big) },\nonumber\\
 \bar{y} = \frac{ (\alpha(x,y) - \beta(x,y))\big(
 y\big(x + \theta_{0}^{1} - \theta_{0}^{2}\big) - \theta_{1}^{2} x\big)
 }{ \alpha(x,y) \big(\theta_{0}^{1} - \theta_{0}^{2}\big) }, \label{eq:ex-ST}
\end{gather}
where $\theta_{i}^{j}$ and $\kappa_{i}$ are some parameters and
\begin{gather*}
\alpha(x,y) = \frac{1}{ (x + y) \big(\theta_{1}^{1} - \theta_{1}^{2}\big) } \left(y r_{1}(x,y) + \frac{ x\big(\theta_{0}^{1} r_{1}(x,y) +
r_{2}(x,y)\big) }{ x + \theta_{0}^{1} - \theta_{0}^{2} }\right), \\
\beta(x,y) = \frac{ \big(\big(y + \theta_{0}^{2}\big) r_{1}(x,y) + r_{2}(x,y)\big) }{
 (x + y) \big(\theta_{1}^{1} - \theta_{1}^{2}\big) },\\
 r_{1}(x,y) = \kappa_{1} \kappa_{2} + \kappa_{2} \kappa_{3} + \kappa_{3} \kappa_{1} -
 \big(y - \theta_{1}^{2}\big)\big(x - \theta_{0}^{2}\big) - \theta_{0}^{1} \big(y + \theta_{0}^{2}\big) -
 \theta_{1}^{1} \big(\theta_{0}^{1} + \theta_{0}^{2} + \theta_{1}^{2}\big),\\
 r_{2}(x,y) =\kappa_{1} \kappa_{2} \kappa_{3} +
 \theta_{1}^{1}\big(\big(y - \theta_{1}^{2}\big)\big(x - \theta_{0}^{2}\big) + \theta_{0}^{1} \big(y + \theta_{0}^{2}\big)\big).
\end{gather*}

On the other hand, there is a well-known example of a discrete Painlev\'e equation of type d-$P\big(A_{2}^{(1)*}\big)$ that was obtained previously by Grammaticos, Ramani, and Ohta as a deautonomization of a QRT mapping~\cite{GraRamOht:2003:AUDOTAQVADIEATST}. This equation has a~much nicer form
\begin{gather}
(f + g)(\bar{f}+g) =\frac{(g+b_1)(g+b_2)(g+b_3)(g+b_4)}{(g-b_5)(g-b_6)},\nonumber\\
(\bar{f}+g)(\bar{f}+\bar{g}) =\frac{(\bar{f}-{b}_1) (\bar{f}-{b}_2)(\bar{f}-{b}_3) (\bar{f}-{b}_4)}{(\bar{f}+{b}_7 - d)(\bar{f}+{b}_8 - d)},\label{eq:ex-QRT}
\end{gather}
where $b_{1},\dots, b_{8}$ are some parameters and $d = b_{1} + \cdots + b_{8}$. Note also that equation~\eqref{eq:ex-ST} is written in the evolutionary form, and equation~\eqref{eq:ex-QRT} is not.

Then the reasonable question to ask is whether these two equations are, in some sense, the same. This question is very natural, since both equations describe essentially the simplest examples of dynamic, in their respective contexts. Contrary to our original expectations, these two equations turned out to be \emph{equivalent} through an explicit change of variables
\begin{gather*}
 f=\dfrac{x \big(y - \theta_{1}^{1}\big) - \big(\kappa_{1} + \theta_{0}^{2} + \theta_{1}^{1}\big) y}{y + \kappa_{1} + \theta_{0}^{2}},\qquad
 g=\dfrac{x (y + \kappa_{1} + \theta_{0}^{1}) + \big(\theta_{0}^{1} - \theta_{0}^{2}\big)y }{x- \kappa_{1} - \theta_{0}^{2}}
\end{gather*}
that transforms one equation into the other. We also have the explicit identification between the two sets of parameters of the equation. It is clear that this equivalence, and especially the resulting change of variables, are impossible to see directly. On the other hand, both follow very naturally from Sakai's geometric theory, and so we think that this is a good illustration of the power of the geometric approach to the theory of discrete Painlev\'e equations.

The paper is organized as follows. In Section~\ref{sec:the_geometry_of_the_okamoto_surface_of_type_a__2_1} we construct an explicit parameterization of the Okamoto space of initial conditions of type $A_{2}^{(1)*}$, we refer to this parameterization as the \emph{canonical model}. In particular, we introduce very important notions of the \emph{period map} and the \emph{root variable parameterization}. In Section~\ref{sec:the_structure_of_the_symmetry_group_widetilde_w_left_e__6_1_right} we construct the birational representation of the extended affine Weyl symmetry group of $\widetilde{W}\big(E_{6}^{(1)}\big)$ of this surface and explain how to represent translational elements in this group in terms of generators, i.e., elementary reflections and automorphisms. Finally, in Section~\ref{sec:comparison_of_two_discrete_painlev_e_equations_of_type_d_p_a__2_1} we construct the spaces of initial conditions for the two equations we consider and find the isomorphisms between the two. This allows us to establish the correspondence between the parameters of the equations through the period map, and also to compute the translational elements corresponding to both equations. We then represent these translational elements as words in the generators of the symmetry group and show that these words are \emph{conjugate}. This establishes the equivalence of the equations and the conjugation element gives us the necessary change of variables transforming one equation into the other.

\section[The geometry of the Okamoto surface of type $A_{2}^{(1)*}$]{The geometry of the Okamoto surface of type $\boldsymbol{A_{2}^{(1)*}}$}\label{sec:the_geometry_of_the_okamoto_surface_of_type_a__2_1}

In this section we prepare the necessary tools to study discrete Painlev\'e equations that are regularized on the family
of generalized Halphen surfaces of type $A_{2}^{(1)*}$.

\subsection[A canonical model of the Okamoto surface of type $A_{2}^{(1)*}$]{A canonical model of the Okamoto surface of type $\boldsymbol{A_{2}^{(1)*}}$}\label{sub:canonical_model}

According to Sakai's theory, all discrete Painlev\'e equations describe dynamics on a family of rational algebraic surfaces that are obtained by blowing up a complex projective plane at a certain number of (possibly infinitely close) points. From the general theory point of view it is better to consider $\mathbb{P}^{2}$ compactification of $\mathbb{C}^{2}$ blown up at~$9$ points since it includes all cases, but from the dynamical systems point of view it is more natural to consider the birationally equivalent $\mathbb{P}^{1} \times \mathbb{P}^{1}$ compactification of $\mathbb{C}^{2}$ blown up at $8$ points (which excludes the $E_{8}^{(1)}$ surface with $A_{0}^{(1)}$ symmetry that corresponds to the $P_{\text{I}}$ case), and that is what we will do. The classification part reflects constraints on configurations of the blowup points, surfaces of different types correspond to different possible configurations. Within each configuration the points can still move and if we denote by $\mathbf{b} = \{b_{i}\}$ the set of parameters describing the location of points within the confi\-guration, we get the family $\mathcal{X}_\mathbf{{b}} := \operatorname{Bl}_{p_{1},\dots,p_{8}}\big(\mathbb{P}^{1} \times \mathbb{P}^{1}\big)$. This family is called the \emph{Okamoto space of initial conditions} or the \emph{Okamoto surface}, for short. The configuration of blowup points is then encoded in the configuration of the irreducible components of the \emph{anti-canonical divisor}~$-\mathcal{K}_{\mathcal{X}}$ of the surface.

The group of the divisor classes $\operatorname{Cl}(\mathcal{X}) = \operatorname{Div}(\mathcal{X})/\operatorname{P}(\mathcal{X})$ in this case coincides with the \emph{Picard group} (also called the \emph{Picard lattice}) $\operatorname{Pic}(\mathcal{X}) = H^{1}(\mathcal{X},\mathcal{O}_{\mathcal{X}}^{*})$. We have
\begin{gather*}
 \operatorname{Pic}\big(\mathbb{P}^{1} \times \mathbb{P}^{1}\big) = \operatorname{Span}_{\mathbb{Z}}\{\mathcal{H}_{f},\mathcal{H}_{g}\},\qquad
 \operatorname{Pic}(\mathcal{X}) = \operatorname{Span}_{\mathbb{Z}}\{\mathcal{H}_{f},\mathcal{H}_{g}, \mathcal{E}_{1},\dots, \mathcal{E}_{8} \},
\end{gather*}
where $(f,g)$ are the coordinates in the affine $\mathbb{C}^{2}$-chart of $\mathbb{P}^{1} \times \mathbb{P}^{1}$ (the other three charts being $(F=1/f,g)$, $(f,G=1/g)$, and $(F,G)$), $\mathcal{H}_{f}$ is the class of \emph{vertical} and $\mathcal{H}_{g}$ is the class of \emph{horizontal} lines on $\mathbb{P}^{1} \times \mathbb{P}^{1}$ (or the classes of their total transform after the blowup), and $\mathcal{E}_{i}$ is the class of the exceptional divisor $E_{i}$ of the blowup centered at $p_{i}$. This lattice is equipped with the \emph{intersection form} that, on the generators, is given by
\begin{gather*}
\mathcal{H}_{f}\bullet \mathcal{H}_{f} = \mathcal{H}_{g}\bullet \mathcal{H}_{g} = \mathcal{H}_{f}\bullet \mathcal{E}_{i} = \mathcal{H}_{g}\bullet \mathcal{E}_{i} = 0,\qquad \mathcal{H}_{f}\bullet \mathcal{H}_{g} = 1,\qquad \mathcal{E}_{i}\bullet \mathcal{E}_{j} = - \delta_{ij}.
\end{gather*}

The anti-canonical divisor class $-\mathcal{K}_{\mathcal{X}}$ is dual to the canonical divisor class $\mathcal{K}_{\mathcal{X}}$ of the top $2$-form $\omega$; on $\mathbb{P}^{1} \times \mathbb{P}^{1}$ we can take $\omega$ to be, up to a multiplication by some rational function, the standard symplectic form,
\begin{gather*}
 \omega = df\wedge dg = - \frac{dF\wedge dg}{F^{2}} = - \frac{df\wedge dG}{G^{2}} = \frac{dF\wedge dG}{F^{2}G^{2}},
\end{gather*}
and therefore
\begin{gather*}
 (\omega) = - 2 \mathcal{H}_{f} - 2 \mathcal{H}_{g} = \mathcal{K}_{\mathbb{P}^{1} \times \mathbb{P}^{1}}\in \operatorname{Pic}\big(\mathbb{P}^{1} \times \mathbb{P}^{1}\big).
\end{gather*}
Lifting this class to $\mathcal{X}$ and dualizing, we get
\begin{gather*}
 - \mathcal{K}_{\mathcal{X}} = 2 \mathcal{H}_{f} + 2 \mathcal{H}_{g} - \mathcal{E}_{1} - \mathcal{E}_{2} - \mathcal{E}_{3}- \mathcal{E}_{4}- \mathcal{E}_{5}- \mathcal{E}_{6}- \mathcal{E}_{7}- \mathcal{E}_{8} \in
 \operatorname{Pic}(\mathcal{X}).
\end{gather*}
We now impose the additional requirement of Painlev\'e theory that $\mathcal{X}_{\mathbf{b}}$ is the \emph{generalized Halphen surface of index $0$}, i.e., that it has the \emph{unique} anti-canonical divisor $-K_{\mathcal{X}}\in - \mathcal{K}_{\mathcal{X}}$ of \emph{canonical type}, which means that if $-K_{\mathcal{X}} = \sum_{i} m_{i} D_{i}$ is the decomposition of $-K_{\mathcal{X}}$ into irreducible components and $\delta_{i} = [D_{i}]$, then $-\mathcal{K}_{\mathcal{X}}\bullet \delta_{i} = 0$ for all $i$. From the \emph{adjunction formula} it then follows that when the components are rational (genus $0$ curves), $\delta_{i}\bullet \delta_{i} = -2$; one exception is the elliptic case and its singular degenerations (or the $A_{0}^{(1)}$, $A_{0}^{(1)*}$, and $A_{0}^{(1)**}$ surfaces), where $- \mathcal{K}_{\mathcal{X}} = \delta$ and $\delta\bullet \delta = 0$.

The decomposition $-\mathcal{K}_{\mathcal{X}} = \sum_{i} m_{i} \delta_{i}$, as well as the intersection configuration of its irreducible components, can be encoded by an \emph{affine Dynkin diagram} and the type of this diagram is the \emph{type} of the Okamoto surface and the corresponding discrete Painlev\'e equations. In the $A_{2}^{(1)}$ case, we get
 \begin{center}
 \begin{tabular}{ccc}
 $\displaystyle
 \raisebox{-0.5in}{\begin{tikzpicture}[
 elt/.style={circle,draw=black!100,thick, inner sep=0pt,minimum size=2mm}]
 \path ( -1,-1) node (D0) [elt] {}
 ( 1,-1) node (D1) [elt] {}
 ( 0,0.8) node (D2) [elt] {};
 \draw [black,line width=1pt ] (D0) -- (D1) -- (D2) -- (D0);
 \node at ($(D0.south east) + (-0.5,-0.1)$) {$\delta_{0}$};
 \node at ($(D1.south west) + (+0.5,-0.1)$) {$\delta_{1}$};
 \node at ($(D2.north west) + (+0.5,+0.1)$) {$\delta_{2}$};
 \end{tikzpicture}}$ \qquad&\qquad $\displaystyle
 \begin{bmatrix}
 -2 & 1 & 1 \\ 1 & -2 & 1 \\ 1 & 1 & -2
 \end{bmatrix}
 $ \qquad&\qquad
 $- \mathcal{K}_{\mathcal{X}} = \delta_{0} + \delta_{1} + \delta_{2}$\\
 Dynkin diagram $A_{2}^{(1)}$ \qquad&\qquad its Cartan matrix \qquad&\qquad $-\mathcal{K}_{\mathcal{X}}$ decomposition
 \end{tabular}
 \end{center}
Note that here it is convenient to use the sign convention for the Cartan matrix $\mathbf{A}= [\delta_{i}\bullet \delta_{j}]$ that is the opposite to the standard one.

Since in the decomposition $- \mathcal{K}_{\mathcal{X}} = \delta = \delta_{0} + \delta_{1} + \delta_{2}$ we want $\delta_{i}$ to be classes of effective divisors, $\delta_{i} = [D_{i}]$ with $\delta_{i}^{2} = -2$, we can choose, essentially up to relabeling,
\begin{gather}
 D_{0} = H_{f} + H_{g} - E_{1} - E_{2} - E_{3} - E_{4} \in
 \delta_{0} = \mathcal{H}_{f} + \mathcal{H}_{g} - \mathcal{E}_{1} - \mathcal{E}_{2} - \mathcal{E}_{3} - \mathcal{E}_{4},\notag\\
 D_{1} = H_{f} - E_{5} - E_{6} \in \delta_{1} = \mathcal{H}_{f} - \mathcal{E}_{5} - \mathcal{E}_{6},\label{eq:d-roots} \\
 D_{2} = H_{g} - E_{7} - E_{8} \in \delta_{2} = \mathcal{H}_{g} - \mathcal{E}_{7} - \mathcal{E}_{8}.\notag
\end{gather}
In other words, $D_{0}$ is the \emph{proper transform} under the blowup procedure of a curve of bi-deg\-ree~$(1,1)$ passing through the blowup points $p_{1},\dots,p_{4}$, $D_{1}$ is the proper transform of a vertical line passing through the points $p_{5}$ and $p_{6}$, and $D_{2}$ is the proper transform of a horizontal line passing through the points $p_{7}$ and $p_{8}$. \emph{This is exactly how the decomposition of the anti-canonical divisor into irreducible components and their intersection configuration encodes the point configuration of the blowup points}.

Note, however, that there are two \emph{different geometric configurations} related to the algebraic intersection structure given by the $A_{2}^{(1)}$ Dynkin diagram:
\begin{center}
 \begin{tabular}{ccc}
 \begin{tikzpicture}[
 elt/.style={circle,draw=black!100,thick, inner sep=0pt,minimum size=2mm}]
 \path ( -1,-1) node (D0) [elt] {}
 ( 1,-1) node (D1) [elt] {}
 ( 0,0.8) node (D2) [elt] {};
 \draw [black,line width=1pt ] (D0) -- (D1) -- (D2) -- (D0);
 \node at ($(D0.south east) + (-0.5,-0.1)$) {$\delta_{0}$};
 \node at ($(D1.south west) + (+0.5,-0.1)$) {$\delta_{1}$};
 \node at ($(D2.north west) + (+0.5,+0.1)$) {$\delta_{2}$};
 \end{tikzpicture} \qquad&\qquad
 \begin{tikzpicture}[
 elt/.style={circle,draw=black!100, fill=black!100, thick, inner sep=0pt,minimum size=1mm}]
 \draw[thick] (-0.2,0.3) -- (1.2,-1.8);
 \draw[thick] (0.2,0.3) -- (-1.2,-1.8);
 \draw[thick] (-1.5,-1.5) -- (1.5,-1.5);
 \node[elt] at (-0.6,-1.5) {}; \node[elt] at (-0.2,-1.5) {}; \node[elt] at (0.2,-1.5) {}; \node[elt] at (0.6,-1.5) {};
 \node[elt] at (-0.7,-1.05) {}; 
 \node[elt] at (-0.3,-0.45) {};
 \node[elt] at (0.7,-1.05) {}; 
 \node[elt] at (0.3,-0.45) {};
 \end{tikzpicture} \qquad&\qquad
 \begin{tikzpicture}[
 elt/.style={circle,draw=black!100, fill=black!100, thick, inner sep=0pt,minimum size=1mm}]
 \draw[thick] (-0.2,0.3) -- (1.2,-1.8);
 \draw[thick] (0.2,0.3) -- (-1.2,-1.8);
 \draw[thick] (0,-1.8) -- (0,0.3);
 \node[elt] at (0,-0.4) {}; \node[elt] at (0,-0.7) {}; \node[elt] at (0,-1) {}; \node[elt] at (0,-1.3) {};
 \node[elt] at (-0.8,-1.2) {}; 
 \node[elt] at (-0.4,-0.6) {};
 \node[elt] at (0.8,-1.2) {}; 
 \node[elt] at (0.4,-0.6) {};
 \end{tikzpicture} \\
 Dynkin diagram $A_{2}^{(1)}$ \qquad&\qquad $A_{2}^{(1)}$ surface (multiplicative) \qquad&\qquad $A_{2}^{(1)*}$ surface (additive)
 \end{tabular}
\end{center}

We are interested in the additive dynamic given by $A_{2}^{(1)*}$, so we want all of the irreducible components of the anti-canonical divisor to intersect at one point. Then, without the loss of generality (i.e., acting by the group of M\"obius transformations on each of the two $\mathbb{P}^{1}$-factors) we can assume that the component $D_{1}= H_{f} - E_{5} - E_{6}$ under the blowing down map projects to the line $f = \infty$ (and so there are two blowup points $p_{5}(\infty, b_{5})$ and $p_{6}(\infty, b_{6})$ on that line), the component $D_{2} = H_{g} - E_{7} - E_{8}$ projects to the line $g = \infty$ with points $p_{7}(-b_{7},\infty)$ and $p_{8}(-b_{8},\infty)$, and the component $D_{0} = H_{f} + H_{g} - E_{1} - E_{2} - E_{3} - E_{4}$ projects to the line $f + g = 0$, and so $D_{0}\cap D_{1}\cap D_{2} = (\infty,\infty)$.

Thus, we get the following geometric realization of a (family of) surface(s) $\mathcal{X}_{\mathbf{b}}$ of type $A_{2}^{(1)*}$:

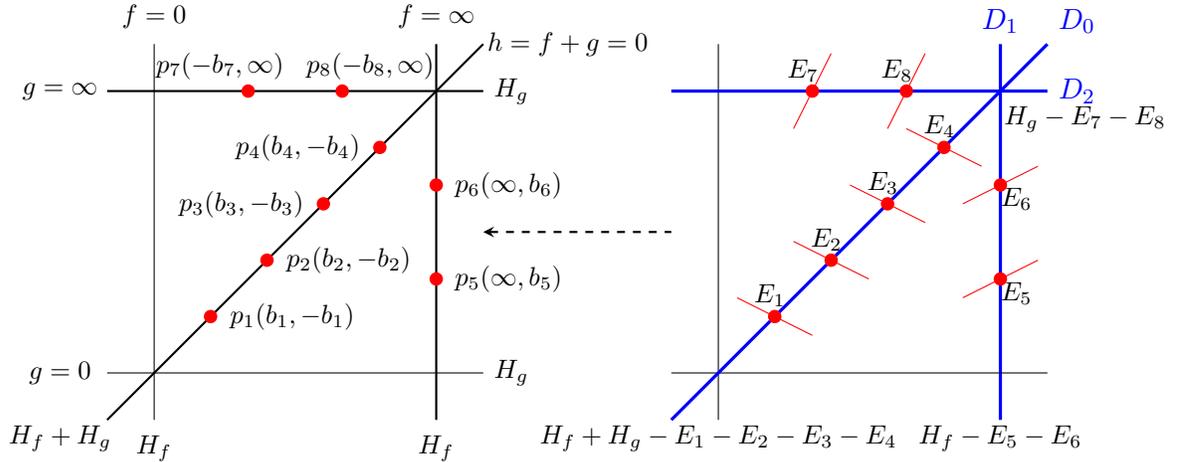
\begin{figure}[h]\centering
\begin{tikzpicture}[>=stealth,
 elt/.style={circle,draw=red!100, fill=red!100, thick, inner sep=0pt,minimum size=1.5mm},scale=1.25]
 \draw[black] (-4.5,-0.5) -- (-4.5,3.5);
 \draw[black,thick] (-1.5,-0.5) -- (-1.5,3.5);
 \draw[black] (-5,0) -- (-1,0);
 \draw[black,thick] (-5,3) -- (-1,3);
 \draw[black,thick] (-5,-0.5) -- (-1,3.5);
 \node[style = elt] (p7) at (-3.5,3) {}; \node[elt] at (-2.5,3) (p8) {};
 \node at ($(p7.north) + (-0.3,0.2)$) {\small $p_{7}(-b_{7},\infty)$};
 \node at ($(p8.north) + (0.3,0.2)$) {\small $p_{8}(-b_{8},\infty)$};
 \node[elt] at (-1.5,2) (p6) {}; \node[elt] at (-1.5,1) (p5) {};
 \node at ($(p5.east) + (0.7,0)$) {\small $p_{5}(\infty,b_{5})$};
 \node at ($(p6.east) + (0.7,0)$) {\small $p_{6}(\infty,b_{6})$};
 \node[elt] at (-2.1,2.4) (p4){}; \node[elt] at (-2.7,1.8) (p3) {};
 \node[elt] at (-3.3,1.2) (p2) {}; \node[elt] at (-3.9,0.6) (p1) {};
 \node at ($(p1.east) + (0.8,0)$) {\small $p_{1}(b_{1},-b_{1})$};
 \node at ($(p2.east) + (0.8,0)$) {\small $p_{2}(b_{2},-b_{2})$};
 \node at ($(p3.west) + (-0.8,0)$) {\small $p_{3}(b_{3},-b_{3})$};
 \node at ($(p4.west) + (-0.8,0)$) {\small $p_{4}(b_{4},-b_{4})$};

 \node at (-5.5,3) {\small $g=\infty$}; \node at (-0.7,3) {\small $H_{g}$};
 \node at (-5.5,0) {\small $g=0$}; \node at (-0.7,0) {\small $H_{g}$};
 \node at (-4.5,3.8) {\small $f=0$}; \node at (-1.5,3.8) {\small $f = \infty$};
 \node at (-4.5,-0.8) {\small $H_{f}$}; \node at (-1.5,-0.8) {\small $H_{f}$};
 \node at (-5.5,-0.7) {\small $H_{f}+H_{g}$};
 \node at (-0.1,3.5) {\small $h = f + g = 0$};
 \draw[black] (1.5,-0.5) -- (1.5,3.5);
 \draw[blue,very thick] (4.5,-0.5) -- (4.5,3.5) node[above] {$D_{1}$};

 \draw[black] (1,0) -- (5,0);
 \draw[blue,very thick] (1,3) -- (5,3) node[right] {$D_{2}$};
 \draw[blue,very thick] (1,-0.5) -- (5,3.5) node[above right] {$D_{0}$};
 \node at (5.4,2.7) {\small $H_{g} - E_{7} - E_{8}$};
 \node at (4.5,-0.7) {\small $H_{f} - E_{5} - E_{6}$};
 \node at (1.5,-0.7) {\small $H_{f} + H_{g} - E_{1} - E_{2} - E_{3} - E_{4}$};

 \node[style = elt] (e7) at (2.5,3) {}; \node[elt] at (3.5,3) (e8) {};
 \draw[red] (2.3,2.6) -- (2.7,3.4); \draw[red] (3.3,2.6) -- (3.7,3.4);
 \node[elt] at (4.5,2) (e6) {}; \node[elt] at (4.5,1) (e5) {};
 \draw[red] (4.1,1.8) -- (4.9,2.2); \draw[red] (4.1,0.8) -- (4.9,1.2);
 \node[elt] at (3.9,2.4) (e4){}; \node[elt] at (3.3,1.8) (e3) {};
 \node[elt] at (2.7,1.2) (e2) {}; \node[elt] at (2.1,0.6) (e1) {};
 \draw[red] (1.7,0.8) -- (2.5,0.4);
 \draw[red] (2.3,1.4) -- (3.1,1);
 \draw[red] (2.9,2) -- (3.7,1.6);
 \draw[red] (3.5,2.6) -- (4.3,2.2);
 \node at ($(e1.north) + (-0.05,0.15)$) {\small $E_{1}$};
 \node at ($(e2.north) + (-0.05,0.15)$) {\small $E_{2}$};
 \node at ($(e3.north) + (-0.05,0.15)$) {\small $E_{3}$};
 \node at ($(e4.north) + (-0.05,0.15)$) {\small $E_{4}$};
 \node at ($(e5.east) + (0.1,-0.15)$) {\small $E_{5}$};
 \node at ($(e6.east) + (0.1,-0.15)$) {\small $E_{6}$};
 \node at ($(e7.north) + (-0.1,0.15)$) {\small $E_{7}$};
 \node at ($(e8.north) + (-0.1,0.15)$) {\small $E_{8}$};
 \draw[thick, black,<-,dashed] (-1,1.5) -- (1,1.5);
\end{tikzpicture}
\caption{Point Configuration and the canonical model of the Okamoto Surface of type $A_{2}^{(1)*}$.}\label{fig:points-standard}
\end{figure}

Note that the points $p_{1},\dots,p_{8}$ lie on the (degenerate) $(2,2)$-curve that is the pole divisor of the 2-form
\begin{gather}\label{eq:symp-form}
 \omega = \frac{df\wedge dg}{f+g} = - \frac{dF\wedge dg}{F(1 + Fg)} = - \frac{df\wedge dG}{G(fG + 1)} = \frac{dF\wedge dG}{FG(F+G)},
\end{gather}
whose pull-back under the blowup gives, after dualizing, the unique anti-canonical divisor $-K_{\mathcal{X}}$.

It is now clear how the parameters $\mathbf{b} = \{b_{1},\dots b_{8}\}$ describe the location of the blowup points within the fixed configuration. However, the true number of parameters if less, since there is still a two-parameter family of M\"obius transformations preserving this configuration:
 \begin{gather*}
 \left(\begin{matrix}
 b_{1} & b_{2} & b_{3} & b_{4}\\
 b_{5} & b_{6} & b_{7} & b_{8}
 \end{matrix}; \begin{matrix}
 f \\ g
 \end{matrix}\right) \sim \left(\begin{matrix}
 \lambda b_{1} + \mu & \lambda b_{2} + \mu & \lambda b_{3} + \mu & \lambda b_{4} + \mu\\
 \lambda b_{5} - \mu & \lambda b_{6} - \mu & \lambda b_{7} - \mu & \lambda b_{8} - \mu
 \end{matrix}; \begin{matrix}
 \lambda f + \mu \\ \lambda g - \mu
 \end{matrix}\right),\qquad \lambda\neq0,
 \end{gather*}
and hence we can use this action to normalize two of the parameters $b_{i}$. It turns out that the correct gauge-invariant parameterization is given by the \emph{root variables} that we now describe.

\subsection{The Period Map and the Root Variable Parameterization}\label{sub:the_period_map_and_the_root_variable_parameterization}
To define the root variables, we first need to introduce the \emph{period map}, which is one of the main tools in the geometric approach, see~\cite{Sak:2001:RSAWARSGPE}. To define this map, we first need to define the \emph{symmetry sub-lattice} $Q = \Pi\big(R^{\perp}\big) \triangleleft \operatorname{Pic}(\mathcal{X})$, which is the orthogonal complement, w.r.t.~the intersection form, of the surface sub-lattice $\Pi(R) = \operatorname{Span}_{\mathbb{Z}}\{\delta_{i}\}\triangleleft \left(- \mathcal{K}_{\mathcal{X}}\right)^{\perp}\triangleleft \operatorname{Pic}(\mathcal{X})$. Note that the lattice bases $R=\{\delta_{i}\}$ and $R^{\perp}=\{\alpha_{j}\}$ can be chosen to consist of \emph{simple roots}, in the sense of the affine Weyl group theory, and so both bases can be encoded by affine Dynkin diagrams. A straightforward
direct computation gives
\begin{gather*}
 Q = \operatorname{Span}_{\mathbb{Z}}\{
 \alpha_{0},\alpha_{1},\alpha_{2},\alpha_{3},\alpha_{4},\alpha_{5},\alpha_{6}\}= Q\big(E_{6}^{(1)}\big)
 = \big(\operatorname{Span}_{\mathbb{Z}}\{\mathcal{\delta}_{0}, \mathcal{\delta}_{1}, \mathcal{\delta}_{2} \}\big)^{\perp}
 = Q\big(\big(A_{2}^{(1)}\big)^{\perp}\big),
\end{gather*}
where the simple roots $\alpha_{i}$ are given by
\begin{equation}\label{eq:a-roots}
 \raisebox{-40pt}{\begin{tikzpicture}[
 elt/.style={circle,draw=black!100,thick, inner sep=0pt,minimum size=2mm}]
 \path (-2,0) node (a0) [elt] {}
 (-1,0) node (a1) [elt] {}
 ( 0,0) node (a2) [elt] {}
 ( 1,0) node (a3) [elt] {}
 ( 2,0) node (a4) [elt] {}
 ( 0,1) node (a5) [elt] {}
 ( 0,2) node (a6) [elt] {};
 \draw [black,line width=1pt ] (a0) -- (a1) -- (a2) -- (a3) -- (a4) (a2) -- (a5) -- (a6);
 \node at ($(a0.south) + (0,-0.2)$) {$\alpha_{0}$};
 \node at ($(a1.south) + (0,-0.2)$) {$\alpha_{1}$};
 \node at ($(a2.south) + (0,-0.2)$) {$\alpha_{2}$};
 \node at ($(a3.south) + (0,-0.2)$) {$\alpha_{3}$};
 \node at ($(a4.south) + (0,-0.2)$) {$\alpha_{4}$};
 \node at ($(a5.east) + (0.3,0)$) {$\alpha_{5}$};
 \node at ($(a6.east) + (0.3,0)$) {$\alpha_{6}$};
 \end{tikzpicture}} \qquad
 \begin{alignedat}{2}
 \alpha_{0} &= \mathcal{E}_{3} - \mathcal{E}_{4}, &\qquad \alpha_{4} &= \mathcal{E}_{7} - \mathcal{E}_{8},\\
 \alpha_{1} &= \mathcal{E}_{2} - \mathcal{E}_{3}, &\qquad \alpha_{5} &= \mathcal{H}_{g} - \mathcal{E}_{1} - \mathcal{E}_{5},\\
 \alpha_{2} &= \mathcal{E}_{1} - \mathcal{E}_{2}, &\qquad \alpha_{6} &= \mathcal{E}_{5} - \mathcal{E}_{6}.\\
 \alpha_{3} &= \mathcal{H}_{f} - \mathcal{E}_{1} - \mathcal{E}_{7},
 \end{alignedat}
\end{equation}
Note also that $\displaystyle \delta = - \mathcal{K}_{\mathcal{X}} = \alpha_{0} + 2 \alpha_{1} + 3 \alpha_{2} + 2 \alpha_{3} + \alpha_{4} + 2 \alpha_{5} + \alpha_{6}$.

The \emph{period map} $\chi\colon Q\to \mathbb{C}$ is defined on the simple roots $\alpha_{i}$ and then extended by the linearity; the \emph{root variables} $a_{i}$ are defined by $a_{i}:= \chi(\alpha_{i})$. To define $\chi(\alpha_{i})$, we proceed as follows, see \cite{Sak:2001:RSAWARSGPE} for details:
\begin{itemize}\itemsep=0pt
\item first, we represent $\alpha_{i}$ as a difference of two effective divisors, $\alpha_{i} = \big[C_{i}^{1}\big] - \big[C_{1}^{0}\big]$;
\item second, there exists a \emph{unique} component $D_{k}$ of $-K_{\mathcal{X}}$ such that $D_{k}\bullet C_{i}^{1} = D_{k}\bullet C_{i}^{0} = 1$; put $P_{i} = D_{k}\cap C_{i}^{0}$ and $Q_{i} = D_{k}\cap C_{i}^{1}$:
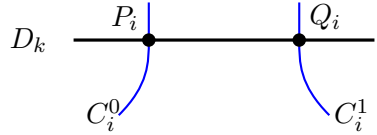
\begin{center}
 \begin{tikzpicture}[>=stealth,
 elt/.style={circle,draw=black!100, fill=black!100, thick, inner sep=0pt,minimum size=1.5mm}]
 \draw[black, very thick] (0,0) -- (4,0);
 \draw[blue, thick] (1,0) -- (1,0.5);
 \draw[blue,thick] (1,0) .. controls (1,-0.3) and (1,-0.6) .. (0.6,-1);
 \draw[blue, thick] (3,0) -- (3,0.5);
 \draw[blue,thick] (3,0) .. controls (3,-0.3) and (3,-0.6) .. (3.4,-1);
 \node[style=elt] (P) at (1,0) {}; \node [above left] at (P) {$P_{i}$};
 \node[style=elt] (Q) at (3,0) {}; \node [above right] at (Q) {$Q_{i}$};
 \node at (-0.6,0) {$D_{k}$};
 \node at (0.4,-1) {$C_{i}^{0}$}; \node at (3.7,-1) {$C_{i}^{1}$};
 \end{tikzpicture}
 \end{center}
 \item then
 \begin{gather*}
 \chi(\alpha_{i}) = \chi\left(\big[C_{i}^{1}\big] - \big[C_{1}^{0}\big]\right) =
 \int_{P_{i}}^{Q_{i}} \frac{ 1 }{ 2 \pi \mathfrak{i} }\oint_{D_{k}} \omega
 = \int_{P_{i}}^{Q_{i}} \operatorname{res}_{D_{k}} \omega,
 \end{gather*}
 where $\omega$ is the symplectic form defined by~\eqref{eq:symp-form}.
 \end{itemize}

\begin{Proposition} For our canonical model of the $A_{2}^{(1)*}$ surface, the period map and the root variables $a_{i} = \chi(\alpha_{i})$ are given by
 \begin{alignat}{4}\label{eq:root-params}
 a_{0} &= b_{4} - b_{3}, &\qquad a_{1} &= b_{3} - b_{2}, &\qquad a_{2} &= b_{2} - b_{1}, &\qquad a_{3} &= b_{1} + b_{7},\\
 a_{4} &= b_{8} - b_{7}, &\qquad a_{5} &= b_{1} + b_{5}, &\qquad a_{6} &= b_{6} - b_{5}.\notag
 \end{alignat}
 This gives us the following parameterization by the root variables $a_{i}$ $($using $a_{i\cdots j}:= a_{i}+\cdots + a_{j})$:
 \begin{equation*}
 \left(\begin{matrix}
 b_{1} & b_{2} & b_{3} & b_{4}\\
 b_{5} & b_{6} & b_{7} & b_{8}
 \end{matrix}; \begin{matrix}
 f \\ g
 \end{matrix}\right) = \left(\begin{matrix}
 b_{4} - a_{012} & b_{4} - a_{01} & b_{4} - a_{0} & b_{4} \\
 a_{0125} - b_{4} & a_{01256} - b_{4} & a_{01233} - b_{4} & a_{01234} - b_{4}
 \end{matrix}; \begin{matrix}
 f \\ g
 \end{matrix}\right).
 \end{equation*}
\end{Proposition}

\begin{Remark} We see that $b_{4}$ is one free parameter (translation of the origin, the choice of~$b_{4}$ turns out to be particularly convenient for our example), the other parameter is the global scaling, that can be used to normalize
\begin{gather*}
d = \chi(\delta) = \chi(- \mathcal{K}_{\mathcal{X}}) = \chi(a_{0} + 2 a_{1} + 3 a_{2} + 2 a_{3} + a_{4} + 2 a_{5} + a_{6})\\
\hphantom{d = \chi(\delta) = \chi(- \mathcal{K}_{\mathcal{X}})}{} = b_{1} + b_{2} + b_{3} + b_{4} + b_{5} + b_{6} + b_{7} + b_{8}.
\end{gather*}
The usual normalization is to put $\chi(\delta) = 1$, and one can also ask the same for~$b_{4}$. We will not do that, but we will require that, when resolving the normalization ambiguity, both $\chi(\delta)$ and~$b_{4}$ are fixed~-- this ensures the group structure on the level of \emph{elementary birational maps}.
 \end{Remark}

\begin{proof} First, note that if we put $h=f+g$, then for the symplectic form $ \omega = \frac{df \wedge dg}{f+g}$ in \eqref{eq:symp-form} we get
 \begin{gather*}
 \operatorname{res}_{D_{0}} \omega = \operatorname{res}_{h=0}\frac{ dh\wedge dg }{ h } = dg, \qquad
 \operatorname{res}_{D_{1}} \omega = \operatorname{res}_{F=0}\frac{ - dF\wedge dg }{ F(1 + FG) } = -dg,\\
 \operatorname{res}_{D_{2}} \omega = \operatorname{res}_{G=0}\frac{ -df\wedge dG }{ G(fG + 1) } = df.
 \end{gather*}

Then
 \begin{gather*}
 \chi(\alpha_{0}) = \chi([E_{3}] - [E_{4}]) =\int_{p_{4}}^{p_{3}}\operatorname{res}_{D_{0}} \omega
 = \int_{-b_{4}}^{ - b_{3}}dg = b_{4} - b_{3} = a_{0},\\
 \chi(\alpha_{1}) = \chi([E_{2}] - [E_{3}]) = b_{3} - b_{2} = a_{1},\qquad
 \chi(\alpha_{2}) = \chi([E_{1}] - [E_{2}]) = b_{2} - b_{1} = a_{2},\\
 \chi(\alpha_{3}) = \chi([H_{f} - E_{1}] - [E_{7}]) =
 \int^{(f = b_{1})\cap D_{2}}_{p_{7}}\operatorname{res}_{D_{2}} \omega = \int^{b_{1}}_{ - b_{7}}df = b_{1} + b_{7} = a_{3},\\
 \chi(\alpha_{4}) = \chi([E_{7}] - [E_{8}]) = \int_{-b_{8}}^{-b_{7}} df = b_{8} - b_{7} = a_{4},\\
 \chi(\alpha_{5}) = \chi([H_{g} - E_{1}] - [E_{5}]) =
 \int^{(g = -b_{1})\cap D_{1}}_{b_{5}}\operatorname{res}_{D_{1}} \omega = \int^{-b_{1}}_{ b_{5}}-dg = b_{1} + b_{5} = a_{5},\\
 \chi(\alpha_{6}) = \chi([E_{5}] - [E_{6}]) = \int_{b_{5}}^{b_{6}} -dg = b_{6} - b_{5} = a_{6}.\tag*{\qed}
 \end{gather*}\renewcommand{\qed}{}
\end{proof}

\section[The structure of the symmetry group $\widetilde{W}\big(E_{6}^{(1)}\big)$]{The structure of the symmetry group $\boldsymbol{\widetilde{W}\big(E_{6}^{(1)}\big)}$} \label{sec:the_structure_of_the_symmetry_group_widetilde_w_left_e__6_1_right}

The next step in understanding the structure of difference Painlev\'e equations of type d-$P\big(A_{2}^{(1)*}\big)$ is to construct the birational representation of the symmetry group of the $A_{2}^{(1)*}$-Okamoto surface. This group is an extended affine Weyl group $\widetilde{W}\big(E_{6}^{(1)}\big) = \operatorname{Aut}\big(E_{6}^{(1)}\big) \ltimes W\big(E_{6}^{(1)}\big)$ that is a \emph{semi-direct product} of the usual affine Weyl group $W\big(E_{6}^{(1)}\big)$ and the group of Dynkin diagram automorphisms $\operatorname{Aut}\big(E_{6}^{(1)}\big) \simeq \operatorname{Aut}\big(A_{2}^{(1)}\big)$ that reflects ambiguities in the choice of the root bases.

\subsection[The affine Weyl group $W\big(E_{6}^{(1)}\big)$]{The affine Weyl group $\boldsymbol{W\big(E_{6}^{(1)}\big)}$}\label{sub:the_affine_weyl_group_w_left_e__6_1_right}

The affine Weyl group $W\big(E_{6}^{(1)}\big)$ is defined in terms of generators $w_{i} = w_{\alpha_{i}}$ and relations that are encoded by the affine Dynkin diagram $E_{6}^{(1)}$,
\begin{gather*}
 W\big(E_{6}^{(1)}\big) = W\left(\raisebox{-25pt}{\begin{tikzpicture}[
 elt/.style={circle,draw=black!100,thick, inner sep=0pt,minimum size=1.5mm}]
 \path (-1,0) node (a0) [elt] {}
 (-0.5,0) node (a1) [elt] {}
 ( 0,0) node (a2) [elt] {}
 ( 0.5,0) node (a3) [elt] {}
 ( 1,0) node (a4) [elt] {}
 ( 0,0.5) node (a5) [elt] {}
 ( 0,1) node (a6) [elt] {};
 \draw [black] (a0) -- (a1) -- (a2) -- (a3) -- (a4) (a2) -- (a5) -- (a6);
 \node at ($(a0.south) + (0,-0.2)$) {$\alpha_{0}$};
 \node at ($(a1.south) + (0,-0.2)$) {$\alpha_{1}$};
 \node at ($(a2.south) + (0,-0.2)$) {$\alpha_{2}$};
 \node at ($(a3.south) + (0,-0.2)$) {$\alpha_{3}$};
 \node at ($(a4.south) + (0,-0.2)$) {$\alpha_{4}$};
 \node at ($(a5.east) + (0.3,0)$) {$\alpha_{5}$};
 \node at ($(a6.east) + (0.3,0)$) {$\alpha_{6}$};
 \end{tikzpicture}} \right)\\
\hphantom{W\big(E_{6}^{(1)}\big)}{}
 =
\left\langle w_{0},\dots, w_{6}\left|
\begin{alignedat}{2}
w_{i}^{2} = e,\quad w_{i}\circ w_{j} &= w_{j}\circ w_{i}& &\text{ when
\raisebox{-0.08in}{\begin{tikzpicture}[
elt/.style={circle,draw=black!100,thick, inner sep=0pt,minimum size=1.5mm}]
\path  ( 0,0) node  (ai) [elt] {}
 ( 0.5,0)    node  (aj) [elt] {};
\draw [black] (ai) (aj);
 \node at ($(ai.south) + (0,-0.2)$)   {$\alpha_{i}$};
\node at ($(aj.south) + (0,-0.2)$) {$\alpha_{j}$};
\end{tikzpicture}}}\\
  w_{i}\circ w_{j}\circ w_{i} &= w_{j}\circ w_{i}\circ w_{j}& &\text{ when
\raisebox{-0.17in}{\begin{tikzpicture}[
elt/.style={circle,draw=black!100,thick, inner sep=0pt,minimum size=1.5mm}]
 \path  ( 0,0) node  (ai) [elt] {}
 ( 0.5,0)    node  (aj) [elt] {};
 \draw [black] (ai) -- (aj);
\node at ($(ai.south) + (0,-0.2)$)   {$\alpha_{i}$};
\node at ($(aj.south) + (0,-0.2)$) {$\alpha_{j}$};
\end{tikzpicture}}}
 \end{alignedat}\right.\right\rangle.
\end{gather*}

There is the natural action of this group on $\operatorname{Pic}(\mathcal{X})$ given by the reflection in the roots~$\alpha_{i}$,
\begin{gather*}
 w_{i}(\mathcal{C}) = w_{\alpha_{i}}(\mathcal{C}) = \mathcal{C} - 2
 \frac{\mathcal{C}\bullet \alpha_{i}}{\alpha_{i}\bullet \alpha_{i}}\alpha_{i}
 = \mathcal{C} + \left(\mathcal{C}\bullet \alpha_{i}\right) a_{i},\qquad \mathcal{C}\in \operatorname{Pic}(\mathcal{X}).
\end{gather*}
Note that each $w_{i}$ is a \emph{Cremona isometry} \cite{Dol:1983:WGACT}, i.e., it preserves the intersection form and the canonical class $\mathcal{K}_{\mathcal{X}}$, and it leaves the semigroup of effective classes invariant. We now extend the action of~$w_{i}$ from $\operatorname{Pic}(\mathcal{X})$ to elementary birational maps, also denoted by~$w_{i}$, of the family~$\mathcal{X}_{\mathbf{b}}$, thus constructing a birational representation of $W\big(E_{6}^{(1)}\big)$.

\begin{Remark}\label{rem:period-action}
In computing the birational representation, the following observation is very helpful. Let $w\in \widetilde{W}\big(E_{6}^{(1)}\big)$, and let $\eta\colon \mathcal{X}_{\mathbf{b}} \to \mathcal{X}_{\bar{\mathbf{b}}}$ be the corresponding mapping, i.e., $w = \eta_{*}$ and $w^{-1} = \eta^{*}$, where $\eta_{*}$ and $\eta^{*}$ are the induced push-forward and pull-back actions on the divisors (and hence on $\operatorname{Pic(\mathcal{X})}$) that are inverses of each other. Since $\eta$ is just a change of the blowdown structure that the period
map $\chi$ does not depend on, $\chi_{\mathcal{X}}(\alpha_{i}) = \chi_{\eta(\mathcal{X})}(\eta_{*}(\alpha_{i}))$. Thus, we can compute the evolution of the root variables directly from the action on $\operatorname{Pic}(\mathcal{X})$ via the formula
\begin{gather*}
 \bar{a}_{i} = \chi_{\eta(\mathcal{X})}(\bar{\alpha}_{i}) = \chi_{\mathcal{X}} (\eta^{*}(\bar{\alpha}_{i}))
 = \chi_{\mathcal{X}} \big(w^{-1}(\bar{\alpha}_{i})\big).
\end{gather*}
Thus \looseness=-1 the action of $\eta$ on the root variables is \emph{inverse} to the action of $w$ on the roots. This is not essential for the generating reflections, that are involutions, but it is important for composed maps.
\end{Remark}

\begin{Theorem}\label{thm:bir-weyl}
Reflections $w_{i}$ on $\operatorname{Pic}(\mathcal{X})$ are induced by the following elementary
 birational mappings, also denoted by $w_{i}$, on the family $\mathcal{X}_{\mathbf{b}}$. To ensure the group structure,
 we require that each map fixes $b_{4}$ and $\chi(\delta)$. We have $($using the notation $b_{i\cdots k} = b_{i} + \cdots + b_{k})$
 \begin{gather*}
 \left(\begin{matrix}
 {b}_{1} & {b}_{2} & {b}_{3} & {b}_{4}\\
 {b}_{5} & {b}_{6} & {b}_{7} & {b}_{8}
 \end{matrix}; \begin{matrix}
 f \\ g
 \end{matrix}\right) \overset{w_{0}}\longmapsto
 \left(\begin{matrix}
 b_{14} - b_{3} & b_{24} - b_{3} & b_{44} - b_{3} & b_{4} \\
 b_{53} - b_{4} & b_{63} - b_{4} & b_{73} - b_{4} & b_{83} - b_{4}
 \end{matrix}; \begin{matrix}
 f - b_{3} + b_{4} \\ g + b_{3} - b_{4}
 \end{matrix}\right), \\
 \left(\begin{matrix}
 {b}_{1} & {b}_{2} & {b}_{3} & {b}_{4}\\
 {b}_{5} & {b}_{6} & {b}_{7} & {b}_{8}
 \end{matrix}; \begin{matrix}
 f \\ g
 \end{matrix}\right) \overset{w_{1}}\longmapsto
 \left(\begin{matrix}
 b_{1} & b_{3} & b_{2} & b_{4} \\
 b_{5} & b_{6} & b_{7} & b_{8}
 \end{matrix}; \begin{matrix}
 f \\ g
 \end{matrix}\right), \\
 \left(\begin{matrix}
 {b}_{1} & {b}_{2} & {b}_{3} & {b}_{4}\\
 {b}_{5} & {b}_{6} & {b}_{7} & {b}_{8}
 \end{matrix}; \begin{matrix}
 f \\ g
 \end{matrix}\right) \overset{w_{2}}\longmapsto
 \left(\begin{matrix}
 {b}_{2} & {b}_{1} & {b}_{3} & {b}_{4}\\
 {b}_{5} & {b}_{6} & {b}_{7} & {b}_{8}
 \end{matrix}; \begin{matrix}
 f \\
 g
 \end{matrix},\right), \\
 \left(\begin{matrix}
 {b}_{1} & {b}_{2} & {b}_{3} & {b}_{4}\\
 {b}_{5} & {b}_{6} & {b}_{7} & {b}_{8}
 \end{matrix}; \begin{matrix}
 f \\ g
 \end{matrix}\right) \overset{w_{3}}\longmapsto
 \left(\begin{matrix}
 -{b}_{7} & {b}_{2} & {b}_{3} & {b}_{4}\\
 {b}_{157} & {b}_{167} & -{b}_{1} & {b}_{8}
 \end{matrix}; \begin{matrix}
 f\\ \frac{(f + b_{7}) (g + b_{1}) }{ f- b_{1} } + b_{7}
 \end{matrix}\right), \\
 \left(\begin{matrix}
 {b}_{1} & {b}_{2} & {b}_{3} & {b}_{4}\\
 {b}_{5} & {b}_{6} & {b}_{7} & {b}_{8}
 \end{matrix}; \begin{matrix}
 f \\ g
 \end{matrix}\right) \overset{w_{4}}\longmapsto
 \left(\begin{matrix}
 b_{1} & b_{2} & b_{3} & b_{4} \\
 b_{5} & b_{6} & b_{8} & b_{7}
 \end{matrix}; \begin{matrix}
 f \\ g
 \end{matrix}\right), \\
 \left(\begin{matrix}
 {b}_{1} & {b}_{2} & {b}_{3} & {b}_{4}\\
 {b}_{5} & {b}_{6} & {b}_{7} & {b}_{8}
 \end{matrix}; \begin{matrix}
 f \\ g
 \end{matrix}\right) \overset{w_{5}}\longmapsto
 \left(\begin{matrix}
 -{b}_{5} & {b}_{2} & {b}_{3} & {b}_{4}\\
 -{b}_{1} & {b}_{6} & {b}_{157} & {b}_{158}
 \end{matrix}; \begin{matrix}
 \frac{ (f-b_{1})(g-b_{5}) }{ g + b_{1} } - b_{5} \\ g
 \end{matrix}\right), \\
 \left(\begin{matrix}
 {b}_{1} & {b}_{2} & {b}_{3} & {b}_{4}\\
 {b}_{5} & {b}_{6} & {b}_{7} & {b}_{8}
 \end{matrix}; \begin{matrix}
 f \\ g
 \end{matrix}\right) \overset{w_{6}}\longmapsto
 \left(\begin{matrix}
 b_{1} & b_{2} & b_{3} & b_{4} \\
 b_{6} & b_{5} & b_{7} & b_{8}
 \end{matrix}; \begin{matrix}
 f \\ g
 \end{matrix}\right).
 \end{gather*}
\end{Theorem}

\begin{proof} Since $\alpha_{0} = \mathcal{E}_{3} - \mathcal{E}_{4}$, the action of $w_{0}$ on $\operatorname{Pic}(\mathcal{X})$ simply interchanges the corresponding divisors, $\mathcal{E}_{3}\leftrightarrow \mathcal{E}_{4}$, which is the same as interchanging the order of blowups or, equivalently, swapping the parameters $b_{3}$ and $b_{4}$. However, since our normalization should fix $b_{4}$, we need to use the gauge action to ensure that $b_{4}$ is fixed:
\begin{gather*}
 \left(\begin{matrix}
 b_{1} & b_{2} & b_{4} & b_{3} \\
 b_{5} & b_{6} & b_{7} & b_{8}
 \end{matrix}; \begin{matrix}
 f \\ g
 \end{matrix}\right)\\
\qquad{} \sim
 \left(\begin{matrix}
 b_{1} - b_{3} + b_{4} & b_{2} - b_{3} + b_{4} & b_{4} - b_{3} + b_{4} & b_{4} \\
 b_{5} + b_{3} - b_{4} & b_{6} + b_{3} - b_{4} & b_{7} + b_{3} - b_{4} & b_{8} + b_{3} - b_{4}
 \end{matrix}; \begin{matrix}
 f - b_{3} + b_{4} \\ g + b_{3} - b_{4}
 \end{matrix}\right).
\end{gather*}

The mapping $w_{1}$ corresponding to the root $\alpha_{1} = \mathcal{E}_{2} - \mathcal{E}_{3}$ interchanges the corresponding divisors, $\mathcal{E}_{2}\leftrightarrow \mathcal{E}_{3}$, which is the same as interchanging the order of blowups or, equivalently, swapping the parameters $b_{2}$ and $b_{3}$. The mappings $w_{2}$, $w_{4}$, and $w_{6}$ are similar.

The two remaining mappings corresponding to the root $\alpha_{3} = \mathcal{H}_{f} - \mathcal{E}_{1} - \mathcal{E}_{7}$ and the root $\alpha_{5} = \mathcal{H}_{g} - \mathcal{E}_{1} - \mathcal{E}_{5}$ are the most interesting ones. The action of $w_{3}$ on $\operatorname{Pic}(\mathcal{X})$ is given by
\begin{gather*}
 w_{3}(\mathcal{H}_{f}) = \mathcal{H}_{f}, \qquad
 w_{3}(\mathcal{H}_{g}) = \mathcal{H}_{f} + \mathcal{H}_{g} - \mathcal{E}_{1} - \mathcal{E}_{7},\qquad
 w_{3}(\mathcal{E}_{1}) = \mathcal{H}_{f} - \mathcal{E}_{7},\\
 w_{3}(\mathcal{E}_{7}) = \mathcal{H}_{f} - \mathcal{E}_{1}, \qquad
 w_{3}(\mathcal{E}_{i}) = \mathcal{E}_{i},\quad i\neq 1,7.
\end{gather*}
Thus, we are looking for a mapping $w_{3}\colon \mathcal{X}_{\mathbf{b}}\to \mathcal{X}_{\bar{\mathbf{b}}}$ that is given in the affine
chart $(f,g)$ by a~formula $w_{3}(f,g) = (\bar{f},\bar{g})$ so that
\begin{gather*}
w_{3}^{*}(\mathcal{H}_{\bar{f}}) = \mathcal{H}_{f},\qquad w_{3}^{*}(\mathcal{H}_{\bar{g}}) = \mathcal{H}_{f} + \mathcal{H}_{g} - \mathcal{E}_{1} - \mathcal{E}_{7}.
\end{gather*}
Thus, up to M\"obius transformations, $\bar{f}$ coincides with $f$ and $\bar{g}$ is a coordinate on a pencil of $(1,1)$-curves
passing through points $p_{1}(b_{1},-b_{1})$ and $p_{7}(-b_{7},\infty)$:
\begin{gather*}
 |H_{\bar{g}}| = \big\{A fg + Bf + C g + D = 0 \,|\, {-}Ab_{1}^{2} + (B-C)b_{1} + D = - Ab_{7} + C = 0\big\} \\
\hphantom{|H_{\bar{g}}|}{} = \{A\left((f + b_{7}) g + b_{1}(b_{1} + b_{7})\right) + B(f - b_{1}) = 0\}.
\end{gather*}
Accounting for the M\"obius transformations, we get
\begin{gather*}
\bar{f} = \frac{ A f + B }{ C f + D },\qquad
\bar{g} = \frac{ K (fg + b_{7} g + b_{1}^{2} + b_{1} b_{7}) + L (f - b_{1}) }{M (fg + b_{7} g + b_{1}^{2} + b_{1} b_{7}) + N (f - b_{1}) },
\end{gather*}
where $A,\dots,N$ are some arbitrary constants that can be determined from the action on the divisors. This can be computed directly, but it is more convenient to use Remark~\ref{rem:period-action}. We see that under the action of $w_{3}$, $\bar{a}_{2} = a_{2} + a_{3}$, $\bar{a}_{3} = - a_{3}$, $\bar{a}_{4} = a_{3} + a_{4}$, and $\bar{a}_{i} = a_{i}$ otherwise. Thus, from \eqref{eq:root-params} we get the following evolution of the coordinates $b_{i}$ of the blowup points:
\begin{gather}\label{eq:param-evol}
 \left(\begin{matrix}
 {b}_{1} & {b}_{2} & {b}_{3} & {b}_{4}\\
 {b}_{5} & {b}_{6} & {b}_{7} & {b}_{8}
 \end{matrix}\right)\overset{w_{3}}\longmapsto
\left(\begin{matrix}
 -{b}_{7} & {b}_{2} & {b}_{3} & {b}_{4}\\
 {b}_{5} + b_{1} + b_{7} & {b}_{6} + b_{1} + b_{7} & -{b}_{1} & {b}_{8}
 \end{matrix}\right).
\end{gather}

Since $w_{3}(\mathcal{E}_{5}) = \mathcal{E}_{5}$, $(\bar{f},\bar{g})(\infty,b_{5}) = (\infty,\bar{b}_{5})$, we see that $C=0$. Similarly, from $w_{3}(\mathcal{E}_{8}) = \mathcal{E}_{8}$ we see that $(\bar{f},\bar{g})(-b_{8},\infty) = (-\bar{b}_{8},\infty)$, and so $M=0$. Thus,
\begin{gather*}
\bar{f} = A f + B ,\qquad \bar{g} = K \frac{ (f + b_{7}) g + b_{1} (b_{1} + b_{7}) }{ f - b_{1} } + L.
\end{gather*}
Finally, from $w_{3}(\mathcal{E}_{i}) = \mathcal{E}_{i}$ for $i=2,3$ and \eqref{eq:param-evol} we immediately see that $A = 1$, $B = 0$, $K=1$, $L=b_{1} + b_{7}$, and we get the required formulae for the mapping $w_{3}$:
\begin{gather*} \bar{f} = f ,\qquad \bar{g} = \frac{ (f + b_{7}) g + b_{1} (b_{1} + b_{7}) }{f - b_{1} } + b_{1} + b_{7} = \frac{(f + b_{7}) (g + b_{1}) }{ f- b_{1} } + b_{7}.
\end{gather*}
The mapping $w_{5}$ is computed in the same way.
\end{proof}

\subsection{The group of Dynkin diagram automorphisms}\label{sub:the_group_of_dynkin_diagram_automorphisms}
It is easy to see, especially looking at the Dynkin diagram $A_{2}^{(1)}$, that the group of Dynkin diagram automorphisms is isomorphic to the usual \emph{dihedral group} $\mathbb{D}_{3}= \{e, m_{0}, m_{1}, m_{2}, r, r^{2}\} = \langle m_{0}, r \,|\, m_{0}^{2} = r^{3} = e,
m_{0} r = r^{2} m_{0}\rangle$ of the symmetries of a triangle:
\begin{equation*}
 \operatorname{Aut}\big(E_{6}^{(1)}\big) = \operatorname{Aut}\left(\raisebox{-25pt}{\begin{tikzpicture}[
 elt/.style={circle,draw=black!100,thick, inner sep=0pt,minimum size=1.5mm}]
 \path (-1,0) node (a0) [elt] {}
 (-0.5,0) node (a1) [elt] {}
 ( 0,0) node (a2) [elt] {}
 ( 0.5,0) node (a3) [elt] {}
 ( 1,0) node (a4) [elt] {}
 ( 0,0.5) node (a5) [elt] {}
 ( 0,1) node (a6) [elt] {};
 \draw [black] (a0) -- (a1) -- (a2) -- (a3) -- (a4) (a2) -- (a5) -- (a6);
 \node at ($(a0.south) + (0,-0.2)$) {$\alpha_{0}$};
 \node at ($(a1.south) + (0,-0.2)$) {$\alpha_{1}$};
 \node at ($(a2.south) + (0,-0.2)$) {$\alpha_{2}$};
 \node at ($(a3.south) + (0,-0.2)$) {$\alpha_{3}$};
 \node at ($(a4.south) + (0,-0.2)$) {$\alpha_{4}$};
 \node at ($(a5.east) + (0.3,0)$) {$\alpha_{5}$};
 \node at ($(a6.east) + (0.3,0)$) {$\alpha_{6}$};
 \end{tikzpicture}} \right) =
 \operatorname{Aut}\big(A_{2}^{(1)}\big) = \operatorname{Aut}\left(\raisebox{-25pt}{\begin{tikzpicture}[
 elt/.style={circle,draw=black!100,thick, inner sep=0pt,minimum size=1.5mm}]
 \path (-0.7,0) node (d0) [elt] {}
 ( 0.7,0) node (d1) [elt] {}
 ( 0,1.2) node (d2) [elt] {};
 \draw [black] (d0) -- (d1) -- (d2) -- (d0);
 \node at ($(d0.south west) + (-0.2,0)$) {$\delta_{0}$};
 \node at ($(d1.south east) + (0.2,0)$) {$\delta_{1}$};
 \node at ($(d2.north east) + (0.3,0)$) {$\delta_{2}$};
 \end{tikzpicture}} \right) = \mathbb{D}_{3}.
\end{equation*}

\begin{Proposition}\label{prop:dyn-auto} The action of $\mathbb{D}_{3}$ on the surface sub-lattice $\operatorname{Span}_{\mathbb{Z}}\{\delta_{i}\}$, on the symmetry sub-lattice $\operatorname{Span}_{\mathbb{Z}}\{\alpha_{i}\}$, and on the Picard lattice $\operatorname{Pic}(\mathcal{X})$, in the cycle notation for permutations, is given by
 \begin{align*}
 m_{0} &= (\delta_{1} \delta_{2})
 = (\alpha_{3} \alpha_{5}) (\alpha_{4} \alpha_{6}),\\
 &\qquad \begin{aligned}
 \mathcal{H}_{f}&\to \mathcal{H}_{g}, & \mathcal{E}_{1}&\to \mathcal{E}_{1}, &
 \mathcal{E}_{3}&\to \mathcal{E}_{3}, & \mathcal{E}_{5}&\to \mathcal{E}_{7}, &
 \mathcal{E}_{7}&\to \mathcal{E}_{5}, \\
 \mathcal{H}_{g}&\to \mathcal{H}_{f}, & \mathcal{E}_{2}&\to \mathcal{E}_{2}, &
 \mathcal{E}_{4}&\to \mathcal{E}_{4}, & \mathcal{E}_{6}&\to \mathcal{E}_{8}, &
 \mathcal{E}_{8}&\to \mathcal{E}_{6};
 \end{aligned}\\
 m_{1} &= (\delta_{0} \delta_{2})
 = (\alpha_{0} \alpha_{4}) (\alpha_{1} \alpha_{3}),\\
 &\qquad \begin{aligned}
 \mathcal{H}_{f}&\to \mathcal{H}_{f}, &
 \mathcal{E}_{1}&\to \mathcal{H}_{f}- \mathcal{E}_{2}, &
 \mathcal{E}_{3}&\to \mathcal{E}_{7}, &
 \mathcal{E}_{5}&\to \mathcal{E}_{5}, &
 \mathcal{E}_{7}&\to \mathcal{E}_{3}, \\
 \mathcal{H}_{g}&\to \mathcal{H}_{f} + \mathcal{H}_{g} - \mathcal{E}_{1} - \mathcal{E}_{2}, &
 \mathcal{E}_{2}&\to \mathcal{H}_{f} - \mathcal{E}_{1}, &
 \mathcal{E}_{4}&\to \mathcal{E}_{8}, &
 \mathcal{E}_{6}&\to \mathcal{E}_{6}, &
 \mathcal{E}_{8}&\to \mathcal{E}_{4};
 \end{aligned} \\
 m_{2} &= (\delta_{0} \delta_{1})
 = (\alpha_{0} \alpha_{6}) (\alpha_{1} \alpha_{5}),\\
 &\qquad \begin{aligned}
 \mathcal{H}_{f}&\to \mathcal{H}_{f} + \mathcal{H}_{g} - \mathcal{E}_{1} - \mathcal{E}_{2}, &
 \mathcal{E}_{1}&\to \mathcal{H}_{g}- \mathcal{E}_{2}, &
 \mathcal{E}_{3}&\to \mathcal{E}_{5}, &
 \mathcal{E}_{5}&\to \mathcal{E}_{3}, &
 \mathcal{E}_{7}&\to \mathcal{E}_{7}, \\
 \mathcal{H}_{g}&\to \mathcal{H}_{g}, &
 \mathcal{E}_{2}&\to \mathcal{H}_{g} - \mathcal{E}_{1}, &
 \mathcal{E}_{4}&\to \mathcal{E}_{6}, &
 \mathcal{E}_{6}&\to \mathcal{E}_{4}, &
 \mathcal{E}_{8}&\to \mathcal{E}_{8};
 \end{aligned}\\
 r &= (\delta_{0} \delta_{1} \delta_{2})
 = (\alpha_{0} \alpha_{6} \alpha_{4}) (\alpha_{1} \alpha_{5} \alpha_{3}),\\
 &\qquad \begin{aligned}
 \mathcal{H}_{f}&\to \mathcal{H}_{g}, &
 \mathcal{E}_{1}&\to \mathcal{H}_{g}- \mathcal{E}_{2}, &
 \mathcal{E}_{3}&\to \mathcal{E}_{5}, &
 \mathcal{E}_{5}&\to \mathcal{E}_{7}, &
 \mathcal{E}_{7}&\to \mathcal{E}_{3}, \\
 \mathcal{H}_{g}&\to \mathcal{H}_{f} + \mathcal{H}_{g} - \mathcal{E}_{1} - \mathcal{E}_{2}, &
 \mathcal{E}_{2}&\to \mathcal{H}_{g} - \mathcal{E}_{1}, &
 \mathcal{E}_{4}&\to \mathcal{E}_{6}, &
 \mathcal{E}_{6}&\to \mathcal{E}_{8}, &
 \mathcal{E}_{8}&\to \mathcal{E}_{4};
 \end{aligned}\\
 r^{2} &= (\delta_{0} \delta_{2} \delta_{1})
 = (\alpha_{0} \alpha_{4} \alpha_{6}) (\alpha_{1} \alpha_{3} \alpha_{5}),\\
 &\qquad \begin{aligned}
 \mathcal{H}_{f}&\to \mathcal{H}_{f} + \mathcal{H}_{g} - \mathcal{E}_{1} - \mathcal{E}_{2}, &
 \mathcal{E}_{1}&\to \mathcal{H}_{f}- \mathcal{E}_{2}, &
 \mathcal{E}_{3}&\to \mathcal{E}_{7}, &
 \mathcal{E}_{5}&\to \mathcal{E}_{3}, &
 \mathcal{E}_{7}&\to \mathcal{E}_{5}, \\
 \mathcal{H}_{g}&\to \mathcal{H}_{f} , &
 \mathcal{E}_{2}&\to \mathcal{H}_{f} - \mathcal{E}_{1}, &
 \mathcal{E}_{4}&\to \mathcal{E}_{8}, &
 \mathcal{E}_{6}&\to \mathcal{E}_{4}, &
 \mathcal{E}_{8}&\to \mathcal{E}_{6}.
 \end{aligned}
 \end{align*}
\end{Proposition}

\begin{proof}The proof is immediate after matching the Dynkin diagrams $A_{2}^{(1)}$ and $E_{6}^{(1)}$ as follows:
 \begin{equation*}
 \begin{tikzpicture}[
 elt/.style={circle,draw=black!100,thick, inner sep=0pt,minimum size=1.5mm},scale=1.5,baseline=0.3in]
 \path (-0.7,0) node (d0) [elt] {}
 ( 0.7,0) node (d1) [elt] {}
 ( 0,1.2) node (d2) [elt] {};
 \draw [black] (d0) -- (d1) -- (d2) -- (d0);
 \node at ($(d0.south west) + (-0.2,0)$) {$\delta_{0}$};
 \node at ($(d1.south east) + (0.2,0)$) {$\delta_{1}$};
 \node at ($(d2.north east) + (0.3,0)$) {$\delta_{2}$};
 \end{tikzpicture} \qquad \Longleftrightarrow \qquad
 \begin{tikzpicture}[
 elt/.style={circle,draw=black!100,thick, inner sep=0pt,minimum size=1.5mm},scale=0.9,baseline=0.1in]
 \path (-2,-1) node (a6) [elt] {}
 (-1,-0.5) node (a5) [elt] {}
 ( 0,0) node (a2) [elt] {}
 ( 1,-0.5) node (a1) [elt] {}
 ( 2,-1) node (a0) [elt] {}
 ( 0,1) node (a3) [elt] {}
 ( 0,2) node (a4) [elt] {};
 \draw [black] (a0) -- (a1) -- (a2) -- (a3) -- (a4) (a2) -- (a5) -- (a6);
 \node at ($(a6.west) + (-0.35,0.1)$) {$\alpha_{6}$};
 \node at ($(a5.west) + (-0.35,0.1)$) {$\alpha_{5}$};
 \node at ($(a2.east) + (0.3,0.1)$) {$\alpha_{2}$};
 \node at ($(a1.east) + (0.3,0.1)$) {$\alpha_{1}$};
 \node at ($(a0.east) + (0.3,0.1)$) {$\alpha_{0}$};
 \node at ($(a3.east) + (0.3,0)$) {$\alpha_{3}$};
 \node at ($(a4.east) + (0.3,0)$) {$\alpha_{4}$};
 \end{tikzpicture},
 \end{equation*}
and using expressions \eqref{eq:d-roots} and \eqref{eq:a-roots} for the simple roots $\delta_{i}$ and $\alpha_{i}$. For example, the reflection~$m_{0}$ should transpose the roots $\delta_{1} = \mathcal{H}_{f} - \mathcal{E}_{5} - \mathcal{E}_{6}$ and $\delta_{2} = \mathcal{H}_{g} - \mathcal{E}_{7} - \mathcal{E}_{8}$ keeping the root $\delta_{0} = \mathcal{H}_{f} + \mathcal{H}_{g} - \mathcal{E}_{1} - \mathcal{E}_{2} - \mathcal{E}_{3} - \mathcal{E}_{4}$ fixed. Fixing~$\mathcal{E}_{i}$ for $i=1,\dots,4$ fixes the roots~$\alpha_{0}$, $\alpha_{1}$, $\alpha_{2}$, thus~$m_{0}$ should transpose the roots $\alpha_{3}$ and $\alpha_{5}$, and also transpose the roots $\alpha_{4}$ and $\alpha_{6}$, which is achieved by transposing $\mathcal{H}_{f}$ and $\mathcal{H}_{g}$, $\mathcal{E}_{5}$ and $\mathcal{E}_{7}$, and $\mathcal{E}_{6}$ and $\mathcal{E}_{8}$. Other cases are similar.
\end{proof}

\begin{Theorem} The action of $\mathbb{D}_{3}$ on $\operatorname{Pic}(\mathcal{X})$ is induced by the following elementary birational mappings on the family $\mathcal{X}_{\mathbf{b}}$ fixing $b_{4}$ and $\chi(\delta)$ $($where we use the notation $b_{i\cdots j}:= b_{i} +\cdots +b_{j})$:
 \begin{gather*}
 \left(\begin{matrix}
 {b}_{1} & {b}_{2} & {b}_{3} & {b}_{4}\\
 {b}_{5} & {b}_{6} & {b}_{7} & {b}_{8}
 \end{matrix}; \begin{matrix}
 f \\ g
 \end{matrix}\right) \overset{m_{0}\,}\longmapsto
 \left(\begin{matrix}
 b_{1} & b_{2} & b_{3} & b_{4} \\
 b_{7} & b_{8} & b_{5} & b_{6}
 \end{matrix}; \begin{matrix}
 -g \\ -f
 \end{matrix}\right), \\
 \left(\begin{matrix}
 {b}_{1} & {b}_{2} & {b}_{3} & {b}_{4}\\
 {b}_{5} & {b}_{6} & {b}_{7} & {b}_{8}
 \end{matrix}; \begin{matrix}
 f \\ g
 \end{matrix}\right)\\
\qquad{} \overset{m_{1}\,}\longmapsto
 \left(\begin{matrix}
 {b}_{4} - b_{28} & {b}_{4} - b_{18} & {b}_{47} - b_{8} & {b}_{4}\\
 {b}_{1258} - b_{4} & {b}_{1268} - b_{4} & {b}_{38} - b_{4} & {b}_{8}
 \end{matrix}; \begin{matrix}
 - f + b_{4} - b_{8} \\ \frac{ f(g + b_{1}) + b_{2}(f - b_{1}) }{ f + g } + b_{8} - b_{4}
 \end{matrix}\right), \\
 \left(\begin{matrix}
 {b}_{1} & {b}_{2} & {b}_{3} & {b}_{4}\\
 {b}_{5} & {b}_{6} & {b}_{7} & {b}_{8}
 \end{matrix}; \begin{matrix}
 f \\ g
 \end{matrix}\right) \\
\qquad{} \overset{m_{2}\,}\longmapsto
 \left(\begin{matrix}
 {b}_{4} - b_{26} & {b}_{4} - b_{16} & {b}_{45} - b_{6} & {b}_{4}\\
 {b}_{36} - b_{4} & {b}_{6} & {b}_{1267} - b_{4} & {b}_{1268} - b_{4}
 \end{matrix}; \begin{matrix}
 \frac{ g(f - b_{1}) - b_{2}(g + b_{1}) }{ f + g } + b_{4} - b_{6} \\ - g - b_{4} + b_{6}
 \end{matrix}\right), \\
 \left(\begin{matrix}
 {b}_{1} & {b}_{2} & {b}_{3} & {b}_{4}\\
 {b}_{5} & {b}_{6} & {b}_{7} & {b}_{8}
 \end{matrix}; \begin{matrix}
 f \\ g
 \end{matrix}\right)\\
\qquad{} \overset{r}\longmapsto
 \left(\begin{matrix}
 {b}_{4} - b_{28} & {b}_{4} - b_{18} & {b}_{47} - b_{8} & {b}_{4}\\
 {b}_{38} - b_{4} & {b}_{8} & {b}_{1258} - b_{4} & {b}_{1268} - b_{4}
 \end{matrix}; \begin{matrix}
 - \frac{ f(g + b_{1}) + b_{2}(f - b_{1})}{ f + g } + b_{4} - b_{8}\\ f - b_{4} + b_{8}
 \end{matrix}\right), \\
 \left(\begin{matrix}
 {b}_{1} & {b}_{2} & {b}_{3} & {b}_{4}\\
 {b}_{5} & {b}_{6} & {b}_{7} & {b}_{8}
 \end{matrix}; \begin{matrix}
 f \\ g
 \end{matrix}\right) \\
\qquad{} \overset{r^{2}}\longmapsto
 \left(\begin{matrix}
 {b}_{4} - b_{26} & {b}_{4} - b_{16} & {b}_{45} - b_{6} & {b}_{4}\\
 {b}_{1267} - b_{4} & {b}_{1268} - b_{4} & {b}_{36} - b_{4} & {b}_{6}
 \end{matrix}; \begin{matrix}
 g + b_{4} - b_{6} \\ - \frac{ g(f - b_{1}) - b_{2}(g + b_{1}) }{ f + g } - b_{4} + b_{6}
 \end{matrix}\right).
 \end{gather*}
\end{Theorem}

\begin{proof} Proof of this theorem is similar to proof of Theorem~\ref{thm:bir-weyl} and is omitted.
\end{proof}

\subsection{The semi-direct product structure}\label{sub:the_semi_direct_product_structure}
The extended affine Weyl group $\widetilde{W}\big(E_{6}^{(1)}\big)$ is a semi-direct product of its normal subgroup
$W\big(E_{6}^{(1)}\big) \triangleleft \widetilde{W}\big(E_{6}^{(1)}\big)$ and the subgroup of the diagram
automorphisms $\operatorname{Aut}\big(E_{6}^{(1)}\big)$,
\begin{gather*}
\widetilde{W}\big(E_{6}^{(1)}\big) = \operatorname{Aut}\big(D_{6}^{(1)}\big) \ltimes W\big(D_{6}^{(1)}\big).
\end{gather*}

We have just described the group structure of $W\big(E_{6}^{(1)}\big)$ and $\operatorname{Aut}\big(E_{6}^{(1)}\big)$ using generators and relations, so it remains to give the action of $\operatorname{Aut}\big(E_{6}^{(1)}\big)$ on $W\big(E_{6}^{(1)}\big)$. But elements of $\operatorname{Aut}\big(E_{6}^{(1)}\big)$ act as permutations of the simple roots $\alpha_{i}$, and so this action is just the same permutation of the corresponding reflections, $\sigma w_{\alpha_{i}} \sigma^{-1} = w_{\sigma(\alpha_{i})}$. For example, the automorphism $m_{1} = (\alpha_{0}\alpha_{4})(\alpha_{1} \alpha_{3})$ acts on $w_{i}$ as
\begin{gather*}
 m_{1} w_{0} m_{1} = w_{4},\qquad m_{1} w_{4} m_{1} = w_{0},\qquad
 m_{1} w_{1} m_{1} = w_{3},\qquad m_{1} w_{3} m_{1} = w_{1},\\
 m_{1} w_{i} m_{1} = w_{i} \quad \text{otherwise}.
\end{gather*}

\subsection{Decomposition of the translation elements}\label{sub:decomposition_of_the_translation_elements}
Finally, we need an algorithm for representing a translation element of $\widetilde{W}\big(E_{6}^{(1)}\big)$ as a composition of the generators of the group, then the corresponding discrete Painlev\'e equation can be understood as a composition of elementary birational maps. For this, we use the generalization to extended affine Weyl groups of Lemma~3.11 of~\cite{Kac:1985:ILA}.

\begin{Lemma}[reduction lemma]Let $W$ be a Weyl group generated by reflections $w_{i}$ in simple roots $\alpha_{i}$, $w\in W$ is a reduced expression in the generators, and let $l(w)$ be the length of $w$. Then $l(w\circ w_{i})< l(w)$ if and only if $w(\alpha_{i})<0$.
\end{Lemma}
In performing computations, the following remark is very useful.

\begin{Remark}\label{rem:root-action}
 Let us create a vector $\upalpha = \langle \alpha_{0}, \dots, \alpha_{j},\dots, \alpha_{n} \rangle$ of simple roots, and let $w_{i}$ act
 of the simple roots as $w_{i}(\alpha_j) = \alpha_{j} + c_{ij}\alpha_{i}$. Then
\begin{align*}
(w\circ w_{i})(\upalpha) &= \langle w(w_{i}(\alpha_{0})),\dots, w(w_{i}(\alpha_{j})),\dots, w(w_{i}(\alpha_{n}))\rangle \\
&= \langle w(\alpha_{0} + c_{i0}\alpha_{i}),\dots, w(\alpha_{j} + c_{ij}\alpha_{i}),\dots, w(\alpha_{n} + c_{in}\alpha_{i}) \rangle \\
&= \langle w(\alpha_{0}) + c_{i0}w(\alpha_{i}),\dots, w(\alpha_{j}) + c_{ij}w(\alpha_{i}),\dots, w(\alpha_{n}) + c_{in}w(\alpha_{i}) \rangle,
\end{align*}
i.e., $(w\circ w_{i})(\upalpha)$ can be easily computed from $w(\upalpha)$ by acting by $w_{i}$ not on $\alpha_{j}$, but on the entries of the vector
$w(\upalpha)$ and using the same coefficients $c_{ij}$.
\end{Remark}

Extending this lemma to include the Dynkin diagram automorphisms is very straightforward. We illustrate this technique by an example corresponding to the dynamic described by equation~\eqref{eq:ex-QRT}.

\begin{Proposition}\label{prop:decomp-QRT} The mapping $\varphi$~\eqref{eq:ex-QRT} can be written in terms of generators of the symmetry group as
 \begin{gather}\label{eq:decomp-QRT}
 \varphi = r\circ w_{5}\circ w_{2}\circ w_{6}\circ w_{5}\circ w_{3}\circ w_{2}
 \circ w_{4}\circ w_{3}\circ w_{1}\circ w_{2}\circ w_{5}\circ w_{0}\circ w_{1}\circ w_{2}
 \circ w_{6}\circ w_{5}.
 \end{gather}
\end{Proposition}
\begin{proof}Since $\varphi$ acts on $b_{i}$ as $\bar{b}_{i} = b_{i}$ for $i=1,\ldots,6$ and $\bar{b}_{i} = b_{i} - d$ for $i=7,8$, in view of~\eqref{eq:root-params} it acts on root variables as
\begin{gather*}
 \varphi\colon \ \upalpha = (a_{0},a_{1},a_{2},a_{3},a_{4},a_{5},a_{6})\mapsto
(a_{0},a_{1},a_{2}, a_{3} - d,a_{4},a_{5}+ d,a_{6}),
\end{gather*}
where $d = \chi_{\mathcal{X}}(\delta) = a_{0} + 2a_{1} + 3 a_{2} + 2 a_{3} + a_{4} + 2a_{5} + a_{6}$ and $\delta$ is the corresponding null root given by the class of the anti-canonical divisor, $\delta = - \mathcal{K}_{\mathcal{X}} = \alpha_{0} + 2\alpha_{1} + 3 \alpha_{2} + 2 \alpha_{3} + \alpha_{4} + 2\alpha_{5} + \alpha_{6}$, $\varphi_{*}$~then acts on the symmetry root basis as
\begin{gather*}
 \varphi_{*}\colon \ \upalpha = (\alpha_{0},\alpha_{1},\alpha_{2},\alpha_{3},\alpha_{4},\alpha_{5},\alpha_{6})\mapsto
(\alpha_{0},\alpha_{1},\alpha_{2}, \alpha_{3} + \delta,\alpha_{4},\alpha_{5}- \delta,\alpha_{6}).
\end{gather*}

Using Remark~\ref{rem:root-action} and the notation $\alpha_{i\cdots j} = \alpha_{i}+\cdots+ \alpha_{j}$, we get
 \begin{gather*}
 \varphi_{*}(\upalpha) = (\alpha_{0},\alpha_{1},\alpha_{2}, \alpha_{3} + \delta,\alpha_{4},
 {\color{purple}\alpha_{5}- \delta},\alpha_{6}),\\
 \big(\varphi_{*}^{(1)} = \varphi_{*}\circ w_{5}\big)(\upalpha) = (\alpha_{0}, \alpha_{1},
 {\color{purple}\alpha_{25} - \delta}, \alpha_{3} + \delta, \alpha_{4}, \delta - \alpha_{5},
 {\color{purple}\alpha_{56} - \delta}),\\
 \big(\varphi_{*}^{(2)} = \varphi_{*}^{(1)}\circ w_{6}\big)(\upalpha) = (\alpha_{0}, \alpha_{1},
 {\color{purple}\alpha_{25} - \delta}, \alpha_{3} + \delta, \alpha_{4}, \alpha_{6},
 {\delta - \alpha_{56}}),\\
 \big(\varphi_{*}^{(3)} = \varphi_{*}^{(2)}\circ w_{2}\big)(\upalpha) = (\alpha_{0},
 {\color{purple}\alpha_{125} - \delta},
 {\delta - \alpha_{25}}, \alpha_{235}, \alpha_{4}, {\color{purple}\alpha_{256} - \delta},
 {\delta - \alpha_{56}}),\\
 \big(\varphi_{*}^{(4)} = \varphi_{*}^{(3)}\circ w_{1}\big)(\upalpha) = ({\color{purple}\alpha_{0125} - \delta},
 \delta - \alpha_{125},\alpha_{1}, \alpha_{235}, \alpha_{4}, {\color{purple}\alpha_{256} - \delta},
 {\delta - \alpha_{56}}),\\
 \big(\varphi_{*}^{(5)} = \varphi_{*}^{(4)}\circ w_{0}\big)(\upalpha) = (\delta - \alpha_{0125}, \alpha_{0},
 \alpha_{1}, \alpha_{235}, \alpha_{4}, {\color{purple}\alpha_{256} - \delta},
 {\delta - \alpha_{56}}),\\
 \big(\varphi_{*}^{(6)} = \varphi_{*}^{(5)}\circ w_{5}\big)(\upalpha) = (\delta - \alpha_{0125}, \alpha_{0},
 {\color{purple}\alpha_{1256} - \delta}, \alpha_{235}, \alpha_{4}, \delta - \alpha_{256},
 \alpha_{2}),\\
 \big(\varphi_{*}^{(7)} = \varphi_{*}^{(6)}\circ w_{2}\big)(\upalpha) = (\alpha_{12233456},
 {\color{purple}-\alpha_{1223345} },
 \alpha_{01223345}, {\color{purple}-\alpha_{01234}}, \alpha_{4}, \alpha_{1},
 \alpha_{2}),\\
 \big(\varphi_{*}^{(8)} = \varphi_{*}^{(7)}\circ w_{1}\big)(\upalpha) = (\alpha_{6},
 \alpha_{1223345},
 \alpha_{0}, {\color{purple}-\alpha_{01234}}, \alpha_{4}, \alpha_{1},
 \alpha_{2}),\\
 \big(\varphi_{*}^{(9)} = \varphi_{*}^{(8)}\circ w_{3}\big)(\upalpha) = (\alpha_{6},
 \alpha_{1223345},
 {\color{purple} -\alpha_{1234}}, \alpha_{01234}, {\color{purple}-\alpha_{0123}}, \alpha_{1},
 \alpha_{2}),\\
 \big(\varphi_{*}^{(10)} = \varphi_{*}^{(9)}\circ w_{4}\big)(\upalpha) = (\alpha_{6},
 \alpha_{1223345},
 {\color{purple} -\alpha_{1234}}, \alpha_{4}, \alpha_{0123}, \alpha_{1},
 \alpha_{2}),\\
 \big(\varphi_{*}^{(11)} = \varphi_{*}^{(10)}\circ w_{2}\big)(\upalpha) = (\alpha_{6},
 \alpha_{235},
 \alpha_{1234}, {\color{purple} -\alpha_{123}}, \alpha_{0123}, {\color{purple} -\alpha_{234}},
 \alpha_{2}), \\
 \big(\varphi_{*}^{(12)} = \varphi_{*}^{(11)}\circ w_{3}\big)(\upalpha) = (\alpha_{6},
 \alpha_{235},
 \alpha_{4}, \alpha_{123}, \alpha_{0}, {\color{purple} -\alpha_{234}},
 \alpha_{2}),\\
 \big(\varphi_{*}^{(13)} = \varphi_{*}^{(12)}\circ w_{5}\big)(\upalpha) = (\alpha_{6},
 \alpha_{235},
 {\color{purple} -\alpha_{23}}, \alpha_{123}, \alpha_{0}, \alpha_{234},
 {\color{purple} - \alpha_{34}}),\\
 \big(\varphi_{*}^{(14)} = \varphi_{*}^{(13)}\circ w_{6}\big)(\upalpha) = (\alpha_{6},
 \alpha_{235},
 {\color{purple} -\alpha_{23}},\alpha_{123}, \alpha_{0}, \alpha_{2},
 \alpha_{34}),\\
 \big(\varphi_{*}^{(15)} = \varphi_{*}^{(14)}\circ w_{2}\big)(\upalpha) = (\alpha_{6},
 \alpha_{5},
 \alpha_{23},\alpha_{1}, \alpha_{0}, {\color{purple} -\alpha_{3}},
 \alpha_{34}),\\
 \big(\varphi_{*}^{(16)} = \varphi_{*}^{(15)}\circ w_{5}\big)(\upalpha) = (\alpha_{6},
 \alpha_{5},
 \alpha_{2},\alpha_{1}, \alpha_{0}, \alpha_{3}, \alpha_{4}).
 \end{gather*}

Finally, we need to apply Dynkin diagram automorphism,
\begin{gather*}
\big(\varphi_{*}^{(17)}=\varphi_{*}^{(15)}\circ r^{2}\big)(\upalpha)=(\alpha_{0},\alpha_{1}, \alpha_{2},\alpha_{3},\alpha_{4},\alpha_{5},\alpha_{6}).
\end{gather*}
 Thus,
\begin{gather*}
\varphi_{*}\circ w_{5}\circ w_{6}\circ w_{2}\circ w_{1}\circ w_{0}\circ w_{5}\circ w_{2}\circ w_{1}\circ w_{3}
\circ w_{4} \circ w_{2}\circ w_{3}\circ w_{5}\circ w_{6}\circ w_{2}\circ w_{5}\circ r^{2} = \operatorname{id},
 \end{gather*}
 and applying the inverse we have
\begin{gather*}
 \varphi_{*} = r\circ w_{5}\circ w_{2}\circ w_{6}\circ w_{5}\circ w_{3}\circ w_{2}
 \circ w_{4}\circ w_{3}\circ w_{1}\circ w_{2}\circ w_{5}\circ w_{0}\circ w_{1}\circ w_{2}
 \circ w_{6}\circ w_{5},
\end{gather*}
which yields the claim.
\end{proof}

\section[Comparison of two discrete Painlev\'e equations of type $d$-$P\big(A_{2}^{(1)*}\big)$]{Comparison of two discrete Painlev\'e equations\\ of type $\boldsymbol{d}$-$\boldsymbol{P\big(A_{2}^{(1)*}\big)}$} \label{sec:comparison_of_two_discrete_painlev_e_equations_of_type_d_p_a__2_1}

\subsection{The deautonomization example}\label{sub:the_deautonomization_example}
This example, obtained by B.~Grammaticos, A.~Ramani, and Y.~Ohta as an application of the \emph{singularity confinement criterion} to
a deautonomization of a particular QRT mapping, was carefully described in~\cite{GraRamOht:2003:AUDOTAQVADIEATST}. Due to the simplicity
structure of the equation we will refer to it as a~\emph{model example}. Consider a birational map $\varphi\colon \mathbb{P}^{1}\times \mathbb{P}^{1} \dashrightarrow \mathbb{P}^{1}\times \mathbb{P}^{1}$ with parameters $b_{1},\dots, b_{8}$:
\begin{align*}
 \varphi\colon \quad &
 \left(\begin{matrix}
 b_{1} & b_{2} & b_{3} & b_{4}\\
 b_{5} & b_{6} & b_{7} & b_{8}
 \end{matrix}; f,g\right) \mapsto
 \left(\begin{matrix}
 \bar{b}_{1} & \bar{b}_{2} & \bar{b}_{3} & \bar{b}_{4}\\
 \bar{b}_{5} & \bar{b}_{6} & \bar{b}_{7} & \bar{b}_{8}
 \end{matrix}; \bar{f},\bar{g}\right),\\
& d= b_{1} + b_{2} + b_{3} + b_{4} + b_{5} + b_{6} + b_{7} + b_{8},\\
& \bar{b}_{1}= b_{1},\qquad \bar{b}_{3} = b_{3},\qquad \bar{b}_{5} = b_{5} + d,\qquad \bar{b}_{7} = b_{7} - d,\\
& \bar{b}_{2}= b_{2},\qquad \bar{b}_{4} = b_{4},\qquad \bar{b}_{6} = b_{6} + d,\qquad \bar{b}_{8} = b_{8} - d,
\end{align*}
and $\bar{f}$ and $\bar{g}$ are given by equation~\eqref{eq:ex-QRT}:
\begin{gather*}
 (f + g)(\bar{f}+g) =\frac{(g+b_1)(g+b_2)(g+b_3)(g+b_4)}{(g-b_5)(g-b_6)},\\
 (\bar{f}+g)(\bar{f}+\bar{g}) =\frac{(\bar{f}-\bar{b}_1) (\bar{f}-\bar{b}_2)(\bar{f}-\bar{b}_3) (\bar{f}-\bar{b}_4)}{
 (\bar{f}+\bar{b}_7)(\bar{f}+\bar{b}_8)}.
\end{gather*}
The singularity structure of this example is the same as the canonical model given on Fig.~\ref{fig:points-standard}. Using the equation, it is quite straightforward to compute the action $\varphi_{*}$ of this mapping on the $\operatorname{Pic}(\mathcal{X})$,
\begin{gather*}
 \mathcal{H}_{f} \mapsto 6 \mathcal{H}_{f} + 3 \mathcal{H}_{g} -2 \mathcal{E}_{1} -
 2 \mathcal{E}_{2} - 2 \mathcal{E}_{3} - 2 \mathcal{E}_{4} - \mathcal{E}_{5} -
 \mathcal{E}_{6} - 3 \mathcal{E}_{7} - 3 \mathcal{E}_{8},\\
 \mathcal{H}_{g} \mapsto 3 \mathcal{H}_{f} + \mathcal{H}_{g} - \mathcal{E}_{1} - \mathcal{E}_{2} -
 \mathcal{E}_{3} - \mathcal{E}_{4} - \mathcal{E}_{7} - \mathcal{E}_{8},\notag\\
 \mathcal{E}_{1} \mapsto 2 \mathcal{H}_{f} + \mathcal{H}_{g} - \mathcal{E}_{2} - \mathcal{E}_{3} -
 \mathcal{E}_{4} - \mathcal{E}_{7} - \mathcal{E}_{8},\notag\\
 \mathcal{E}_{2} \mapsto 2 \mathcal{H}_{f} + \mathcal{H}_{g} - \mathcal{E}_{1} - \mathcal{E}_{3} -
 \mathcal{E}_{4} - \mathcal{E}_{7} - \mathcal{E}_{8},\notag\\
 \mathcal{E}_{3} \mapsto 2 \mathcal{H}_{f} + \mathcal{H}_{g} - \mathcal{E}_{1} - \mathcal{E}_{2} -
 \mathcal{E}_{4} - \mathcal{E}_{7} - \mathcal{E}_{8},\notag\\
 \mathcal{E}_{4} \mapsto 2 \mathcal{H}_{f} + \mathcal{H}_{g} - \mathcal{E}_{1} - \mathcal{E}_{2} -
 \mathcal{E}_{3} - \mathcal{E}_{7} - \mathcal{E}_{8},\notag\\
 \mathcal{E}_{5} \mapsto 3 \mathcal{H}_{f} + \mathcal{H}_{g} - \mathcal{E}_{1} - \mathcal{E}_{2} -
 \mathcal{E}_{3} - \mathcal{E}_{4} - \mathcal{E}_{6} - \mathcal{E}_{7} - \mathcal{E}_{8},\notag\\
 \mathcal{E}_{6} \mapsto 3 \mathcal{H}_{f} + \mathcal{H}_{g} - \mathcal{E}_{1} - \mathcal{E}_{2} -
 \mathcal{E}_{3} - \mathcal{E}_{4} - \mathcal{E}_{5} - \mathcal{E}_{7} - \mathcal{E}_{8},\notag\\
 \mathcal{E}_{7} \mapsto \mathcal{H}_{f} - \mathcal{E}_{8},\notag\\
 \mathcal{E}_{8} \mapsto \mathcal{H}_{f} - \mathcal{E}_{7}.\notag
\end{gather*}
Thus, the induced action $\varphi_{*}$ on the sub-lattice $R^{\perp}$ is given by the \emph{translation} considered in
Proposition~\ref{prop:decomp-QRT} and given by~\eqref{eq:decomp-QRT},
\begin{gather*}
 \varphi_{*}\colon \ (\alpha_{0}, \alpha_{1}, \alpha_{2}, \alpha_{3}, \alpha_{4}, \alpha_{5}, \alpha_{6})\mapsto
 (\alpha_{0}, \alpha_{1}, \alpha_{2}, \alpha_{3}, \alpha_{4}, \alpha_{5}, \alpha_{6}) + (0,0,0,1,0,-1,0)\delta,\label{eq:dpa-trans-std}
\end{gather*}
as well as the permutation $\sigma_{r} = (\delta_{0}\delta_{1}\delta_{2}) = (D_{0}D_{1}D_{2})$ of the irreducible components of $-K_{\mathcal{X}}$.

\subsection{The Schlesinger transformation example}\label{sub:the_schlesinger_transformation_example}
This example has been described in detail in \cite{DzhSakTak:2013:DSTTHFADPE}, so below we only give a very brief outline of the setup. We consider a Fuchsian system of the spectral type $111,111,111$, i.e., this system has $n=2$ (finite) poles and the matrix size $m=3$. It is possible to map the finite poles to $z_{0} = 0$ and $z_{1}= 1$ by a M\"obius transformation and then use scalar gauge transformations to make
$\operatorname{rank}(\mathbf{A}_{i}) = 2$ at finite poles. Then our Fuchsian system has the form
\begin{gather*}
 \frac{\mathbf{dY}}{dz} = \mathbf{A}(z)\mathbf{Y} = \left(\frac{\mathbf{A}_{0}}{z} + \frac{\mathbf{A}_{1}}{z-1}\right)\mathbf{Y};
\end{gather*}
we also put $\mathbf{A}_{\infty} = - \mathbf{A}_{0} - \mathbf{A}_{1}$. The eigenvalues $\theta_{i}^{j}$ of $\mathbf{A}_{i}$ satisfy the
the \emph{Fuchs relation} and are encoded by the \emph{Riemann scheme}
\begin{gather*}
 \left\{
 \begin{tabular}{cccc}
 $z = 0$ & $ z = 1 $ & $z = \infty$ \\
 $\theta_{0}^{1}$ & $\theta_{1}^{1}$ & $\kappa_{1}$ \\
 $\theta_{0}^{2}$ & $\theta_{1}^{2}$ & $\kappa_{2}$ \\
 $ 0 $ & $ 0$ & $\kappa_{3}$
 \end{tabular}
 \right\},\qquad
 \theta_{0}^{1} + \theta_{0}^{2} + \theta_{1}^{1} + \theta_{1}^{2} + \sum_{j=1}^{3} \kappa_{j}= 0.
\end{gather*}
We consider an elementary two-point Schlesinger transformation $\left\{\begin{smallmatrix} 0&1\\1&1\end{smallmatrix}\right\}$ that acts
on the characteristic indices $\theta_{i}^{j}$ and $\kappa_{i}$ as follows:
\begin{gather*}
\left\{\begin{smallmatrix} 0&1\\1&1\end{smallmatrix}\right\}\colon \quad \bar{\theta}^{1}_{0} = \theta^{1}_{0} - 1, \qquad
\bar{\theta}^{1}_{1} = \theta^{1}_{1} + 1,\qquad \bar{\theta_{i}^{j}} = \theta_{i}^{j}\text{ otherwise}, \qquad \bar{\kappa}_{i} = \kappa_{i}.
\end{gather*}
This transformation can be performed using the specially chosen \emph{multiplier} matrix $\mathbf{R}(z)$ via $\overline{\mathbf{Y}}(z) = \mathbf{R}(z)\mathbf{Y}(z)$. Using the eigenvector decomposition of the coefficient matrices,
\begin{gather*}
 \mathbf{A}_{i} = \mathbf{B}_{i} \mathbf{C}_{i}^{\dag} =
 \begin{bmatrix}
 \mathbf{b}_{i,1} & \mathbf{b}_{i,2}
 \end{bmatrix} \begin{bmatrix}
 \mathbf{c}_{i}^{1\dag}\\[2pt] \mathbf{c}_{i}^{2\dag}
 \end{bmatrix},\qquad
 \mathbf{C}_{i}^{\dag} \mathbf{B}_{i} = \mathbf{\Theta}_{i} = \operatorname{diag}\big\{\theta_{i}^{1},\theta_{i}^{2}\big\},
\end{gather*}
and some remaining gauge freedom, we get the following parameterization:
\begin{gather*}
 \mathbf{B}_{0} = \begin{bmatrix}
 1 & 0 \\ 0 & 1 \\ 0 & 0
 \end{bmatrix}, \qquad \mathbf{C}_{0}^{\dag} = \begin{bmatrix}
 \theta_{0}^{1} & 0 & \alpha \\ 0 & \theta_{0}^{2} & \beta
 \end{bmatrix},\qquad
 \mathbf{B}_{1} = \begin{bmatrix}
 0 & 1 \\ 0 & 1 \\ 1 & 1
 \end{bmatrix}, \qquad
 \mathbf{C}_{1}^{\dag} = \begin{bmatrix}
 -\gamma - \theta_{1}^{1} & \gamma & \theta_{1}^{1} \\ \theta_{1}^{2} - \delta & \delta & 0
 \end{bmatrix}.
\end{gather*}
Requiring that the eigenvalues of $\mathbf{A}_{\infty}$ are $\kappa_{1}$, $\kappa_{2}$, and $\kappa_{3}$ results in the equations $\operatorname{tr} (\mathbf{A}_{\infty}) = \kappa_{1} + \kappa_{2} + \kappa_{3}$ (which is just the Fuchs relation), $|\mathbf{A}_{\infty}|_{11} + |\mathbf{A}_{\infty}|_{22} + |\mathbf{A}_{\infty}|_{33} = \kappa_{2}\kappa_{3} + \kappa_{3}\kappa_{1} + \kappa_{1}\kappa_{2}$ (where $|\mathbf{A}|_{ij}$ denotes the
$(ij)$-minor of $\mathbf{A}$) and $\operatorname{det}(\mathbf{A}_{\infty}) = \kappa_{1} \kappa_{2} \kappa_{3}$. We then notice, see \cite{DzhSakTak:2013:DSTTHFADPE, DzhTak:2015:GAORFSTTDPE} for details, that it is convenient to choose, as our coordinates,
\begin{gather*}
 x = \frac{ (\gamma + \delta)\big(\theta_{0}^{1} - \theta_{0}^{2}\big) }{ \theta_{1}^{1} - \theta_{1}^{2} },
 \qquad
 y = \frac{ \theta_{1}^{2} \gamma + \theta_{1}^{1} \delta }{ \gamma + \delta + \theta_{1}^{1} - \theta_{1}^{2}}.
\end{gather*}
Then
\begin{gather*}
 \alpha(x,y) = \frac{ 1 }{ (x + y) \big(\theta_{1}^{1} - \theta_{1}^{2}\big) } \left(y r_{1}(x,y) + \frac{ x\big(\theta_{0}^{1} r_{1}(x,y) + r_{2}(x,y)\big) }{ x + \theta_{0}^{1} - \theta_{0}^{2} }\right),\\
 \beta(x,y) = \frac{ 1 }{ (x + y) \big(\theta_{1}^{1} - \theta_{1}^{2}\big) } \big( \big(y + \theta_{0}^{2}\big) r_{1}(x,y) + r_{2}(x,y)\big),
\end{gather*}
where
\begin{gather*}\begin{split}
 &r_{1}(x,y) = \kappa_{1} \kappa_{2} + \kappa_{2} \kappa_{3} + \kappa_{3} \kappa_{1} - \big(y - \theta_{1}^{2}\big)\big(x - \theta_{0}^{2}\big) - \theta_{0}^{1} \big(y + \theta_{0}^{2}\big) -\theta_{1}^{1} \big(\theta_{0}^{1} + \theta_{0}^{2} + \theta_{1}^{2}\big),\\
& r_{2}(x,y) = \kappa_{1} \kappa_{2} \kappa_{3} + \theta_{1}^{1}\big(\big(y - \theta_{1}^{2}\big)\big(x - \theta_{0}^{2}\big) + \theta_{0}^{1} \big(y + \theta_{0}^{2}\big)\big).\end{split}
\end{gather*}
For the elementary Schlesinger transformation $\left\{\begin{smallmatrix} 0&1\\1&1\end{smallmatrix}\right\}$ the multiplier
matrix has the form $\mathbf{R}(z) = \mathbf{I} - \frac{1}{z} \frac{\mathbf{b}_{1,1}\mathbf{c}_{0}^{1\dag}}{\mathbf{c}_{0}^{1\dag}\mathbf{b}_{1,1}}$,
and the resulting discrete Schlesinger evolution equations, when written in terms of the eigenvector dynamics,
again see \cite{DzhSakTak:2013:DSTTHFADPE, DzhTak:2015:GAORFSTTDPE} for details, give us the map \eqref{eq:ex-ST}:
$\psi\colon (x,y)\to (\bar{x},\bar{y})$:
\begin{gather*}
 \bar{x}= \frac{(\alpha(x,y) - \beta(x,y)) \big(\alpha(x,y) x \big(\theta_{1}^{1} - \theta_{1}^{2}\big) +
 \big(1 + \theta_{0}^{2}\big)\big(x \big(y - \theta_{1}^{2}\big) + y \big(\theta_{0}^{1} - \theta_{0}^{2}\big)\big)\big) }{
 (\alpha(x,y) - \beta(x,y)) (x \big(y - \theta_{1}^{2}\big) + \big(\theta_{0}^{1} - \theta_{0}^{2}\big)y) -
 \alpha(x,y) \big(\theta_{1}^{1} + 1\big)\big(\theta_{0}^{1} - \theta_{0}^{2}\big) },\\
 \bar{y} = \frac{ (\alpha(x,y) - \beta(x,y))\big(y(x + \theta_{0}^{1} - \theta_{0}^{2}) - \theta_{1}^{2} x\big)
 }{ \alpha(x,y) \big(\theta_{0}^{1} - \theta_{0}^{2}\big) }. 
\end{gather*}
Using a Computer Algebra System, such as \textbf{Mathematica}, we can find and resolve the indeterminate points of the dynamic to obtain the blowup diagram on Fig.~\ref{fig:points-ST}. It is essentially the same as the canonical model, and so we can use the same choice of the root bases. However, now the coordinates of the blow-up points are
\begin{alignat*}{4}
 &p_{1}\big(\theta_{0}^{2} + \kappa_{1}, - \theta_{0}^{2} - \kappa_{1}\big), & \qquad
 &p_{3}\big(\theta_{0}^{2} + \kappa_{3}, - \theta_{0}^{2} - \kappa_{3}\big), &\qquad
 &p_{5}\big(\infty,\theta_{1}^{1}\big),&\qquad
 &p_{7}\big(\theta_{0}^{2} - \theta_{0}^{1},\infty\big),\\
 &p_{2}\big(\theta_{0}^{2} + \kappa_{2}, - \theta_{0}^{2} - \kappa_{2}\big), &\qquad
 &p_{4}(0,0),&\quad
 &p_{6}\big(\infty,\theta_{1}^{2}\big),&\qquad
 &p_{8}\big(\theta_{0}^{2}+1,\infty\big),
\end{alignat*}
and this allows us to perform the parameter matching:
\begin{alignat*}{4}
 b_{1} &= \theta_{0}^{2} + \kappa_{1},&\qquad b_{2} &= \theta_{0}^{2} + \kappa_{2},&\qquad
 b_{3} &= \theta_{0}^{2} + \kappa_{3},&\qquad b_{4} &= 0,\\
 b_{5} &= \theta_{1}^{1},&\qquad b_{6} &= \theta_{1}^{2},&\qquad
 b_{7} &= \theta_{0}^{1} - \theta_{0}^{2},&\qquad b_{8} &= - \theta_{0}^{2} - 1.
\end{alignat*}

Thus, $d = b_{1}+\cdots + b_{8} = -1$, and we get the following root variable evolution:
\begin{alignat*}{3}
 \bar{a}_{i} &= a_{i},\quad i=0,1,2,&\qquad
 \bar{a}_{3} &= a_{3} - 1 = a_{3} + d,&\qquad
 \bar{a}_{4} &= a_{4} + 1 = a_{4} - d,\\
 \bar{a}_{5} &= a_{5} + 1 = a_{5} - d,&\qquad
 \bar{a}_{6} &= a_{6} - 1 = a_{6} + d.
\end{alignat*}

 \begin{figure}[ht]
 \centering
 \begin{tikzpicture}[>=stealth,
 elt/.style={circle,draw=red!100, fill=red!100, thick, inner sep=0pt,minimum size=1.5mm},scale=1.25]
 \draw[black] (-4.5,-0.5) -- (-4.5,3.5);
 \draw[black,thick] (-1.5,-0.5) -- (-1.5,3.5);
 \draw[black] (-5,0) -- (-1,0);
 \draw[black,thick] (-5,3) -- (-1,3);
 \draw[black,thick] (-5,-0.5) -- (-1,3.5);
 \node[style = elt] (p7) at (-3.5,3) {}; \node[elt] at (-2.5,3) (p8) {};
 \node at ($(p7.north) + (0,0.2)$) {\small $p_{7}$};
 \node at ($(p8.north) + (0,0.2)$) {\small $p_{8}$};
 \node[elt] at (-1.5,2) (p6) {}; \node[elt] at (-1.5,1) (p5) {};
 \node at ($(p5.east) + (0.2,0)$) {\small $p_{5}$};
 \node at ($(p6.east) + (0.2,0)$) {\small $p_{6}$};
 \node[elt] at (-4.5,0) (p4){}; \node[elt] at (-2.2,2.3) (p1) {};
 \node[elt] at (-2.8,1.7) (p2) {}; \node[elt] at (-3.4,1.1) (p3) {};
 \node at ($(p1.east) + (0.2,0)$) {\small $p_{1}$};
 \node at ($(p2.east) + (0.2,0)$) {\small $p_{2}$};
 \node at ($(p3.east) + (0.2,0)$) {\small $p_{3}$};
 \node at ($(p4.east) + (0.2,-0.2)$) {\small $p_{4}$};

 \node at (-5.5,3) {\small $y=\infty$}; \node at (-0.7,3) {\small $H_{y}$};
 \node at (-5.5,0) {\small $y=0$}; \node at (-0.7,0) {\small $H_{y}$};
 \node at (-4.5,3.8) {\small $x=0$}; \node at (-1.5,3.8) {\small $x = \infty$};
 \node at (-4.5,-0.8) {\small $H_{x}$}; \node at (-1.5,-0.8) {\small $H_{x}$};
 \node at (-5.5,-0.7) {\small $H_{x}+H_{y}$};
 \node at (-0.1,3.5) {\small $x + y = 0$};

 \draw[black] (1,0.5) -- (2.7,3.5);
 \draw[blue,very thick] (4.5,-0.5) -- (4.5,3.5) node[above] {$D_{1}$};

 \draw[black] (2,-0.5) -- (5,1.2);
 \draw[blue,very thick] (1,3) -- (5,3) node[right] {$D_{2}$};
 \draw[blue,very thick] (1,-0.5) -- (5,3.5) node[above right] {$D_{0}$};
 \node at (5.4,2.7) {\small $H_{y} - E_{7} - E_{8}$};
 \node at (4.5,-0.7) {\small $H_{x} - E_{5} - E_{6}$};
 \node at (1.5,-0.7) {\small $H_{x} + H_{g} - E_{1} - E_{2} - E_{3} - E_{4}$};

 \node[style = elt] (e7) at (3,3) {}; \node[elt] at (3.5,3) (e8) {};
 \draw[red] (2.8,2.6) -- (3.2,3.4); \draw[red] (3.3,2.6) -- (3.7,3.4);
 \node[elt] at (4.5,2) (e6) {}; \node[elt] at (4.5,1.5) (e5) {};
 \draw[red] (4.1,1.8) -- (4.9,2.2);
 \draw[red] (4.1,1.3) -- (4.9,1.7);
 \node[elt] at (3.9,2.4) (e4){}; \node[elt] at (3.3,1.8) (e3) {};
 \node[elt] at (2.7,1.2) (e2) {}; \node[elt] at (2.1,0.6) (e1) {};
 \draw[red] (1.1,1.6) -- (3.1,-0.4);
 \draw[red] (2.3,1.4) -- (3.1,1);
 \draw[red] (2.9,2) -- (3.7,1.6);
 \draw[red] (3.5,2.6) -- (4.3,2.2);
 \node at ($(e1.south) + (0.05,-0.15)$) {\small $E_{4}$};
 \node at ($(e2.north) + (-0.05,0.15)$) {\small $E_{3}$};
 \node at ($(e3.north) + (-0.05,0.15)$) {\small $E_{2}$};
 \node at ($(e4.north) + (-0.05,0.15)$) {\small $E_{1}$};
 \node at ($(e5.east) + (0.1,-0.15)$) {\small $E_{5}$};
 \node at ($(e6.east) + (0.1,-0.15)$) {\small $E_{6}$};
 \node at ($(e7.north) + (-0.1,0.15)$) {\small $E_{7}$};
 \node at ($(e8.north) + (-0.1,0.15)$) {\small $E_{8}$};

 \draw[thick, black,<-,dashed] (-1,1.5) -- (1,1.5);
 \end{tikzpicture}
\caption{Point Configuration for the Schlesinger transformation example.} \label{fig:points-ST}
\end{figure}
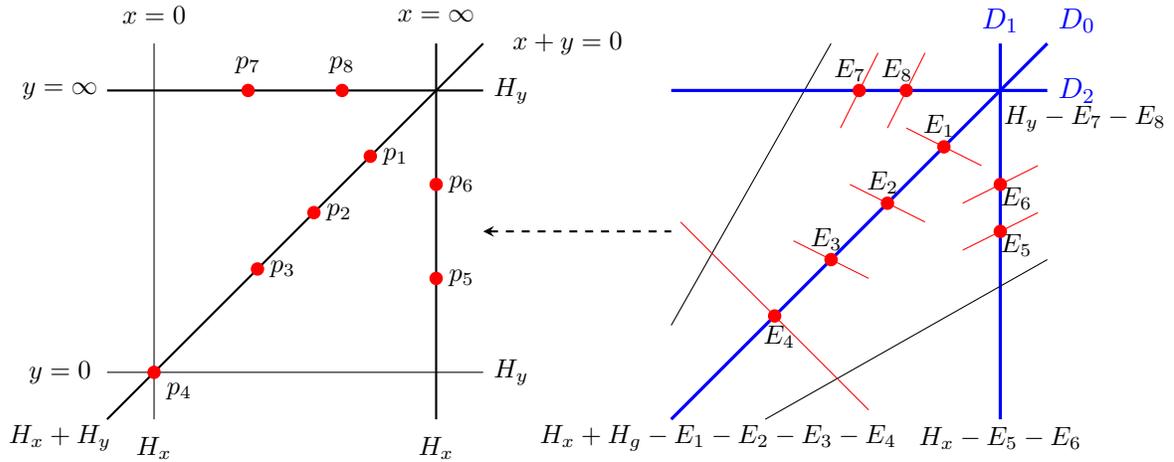

Using Remark~\ref{rem:period-action}, we get the following action of the mapping $\psi_{*}\colon \operatorname{Pic}(\mathcal{X})\to \operatorname{Pic}(\mathcal{X})$ on the roots:
\begin{gather*}
\psi_{*}\colon \ (\alpha_{0}, \alpha_{1}, \alpha_{2}, \alpha_{3}, \alpha_{4}, \alpha_{5}, \alpha_{6})\mapsto (\alpha_{0}, \alpha_{1}, \alpha_{2}, \alpha_{3}, \alpha_{4}, \alpha_{5}, \alpha_{6}) + (0,0,0,-1,1,1,-1)\delta.
\end{gather*}

\begin{Proposition}\label{prop:decomp-ST} The mapping $\psi$ can be written in terms of generators of the symmetry group as
\begin{gather}\label{eq:decomp-ST}
\psi = r\circ w_{1}\circ w_{2}\circ w_{6}\circ w_{5}\circ w_{3}\circ w_{2}\circ w_{4}\circ w_{3}\circ w_{1}\circ w_{2}\circ w_{5}\circ w_{0}\circ w_{1}\circ w_{2}\circ w_{6}\circ w_{3}.
\end{gather}
Thus, the mapping $\psi_{*}$ acts on $\operatorname{Pic}(\mathcal{X})$ as
\begin{gather*}
\mathcal{H}_{f} \mapsto 2 \mathcal{H}_{f} + 3 \mathcal{H}_{g} - \mathcal{E}_{1} - \mathcal{E}_{2} - \mathcal{E}_{3} - \mathcal{E}_{4} - 2 \mathcal{E}_{5} - 2\mathcal{E}_{8}, \\
\mathcal{H}_{g} \mapsto 3 \mathcal{H}_{f} + 5 \mathcal{H}_{g} - 2\mathcal{E}_{1} - 2\mathcal{E}_{2} - 2\mathcal{E}_{3} -2\mathcal{E}_{4} - 3 \mathcal{E}_{5} - \mathcal{E}_{6} - 2\mathcal{E}_{8},\\ 
\mathcal{E}_{1} \mapsto \mathcal{H}_{f} + 2\mathcal{H}_{g} - \mathcal{E}_{2} - \mathcal{E}_{3} - \mathcal{E}_{4} - \mathcal{E}_{5} - \mathcal{E}_{8}, \\ 
\mathcal{E}_{2} \mapsto \mathcal{H}_{f} + 2\mathcal{H}_{g} - \mathcal{E}_{1} - \mathcal{E}_{3} - \mathcal{E}_{4} -\mathcal{E}_{5} - \mathcal{E}_{8},\\ 
\mathcal{E}_{3} \mapsto \mathcal{H}_{f} + 2\mathcal{H}_{g} - \mathcal{E}_{1} - \mathcal{E}_{2} - \mathcal{E}_{4} -\mathcal{E}_{5} - \mathcal{E}_{8},\\ 
\mathcal{E}_{4} \mapsto \mathcal{H}_{f} + 2\mathcal{H}_{g} - \mathcal{E}_{1} - \mathcal{E}_{2} - \mathcal{E}_{3} - \mathcal{E}_{5} - \mathcal{E}_{8}, \\
\mathcal{E}_{5} \mapsto \mathcal{E}_{7},\\
\mathcal{E}_{6} \mapsto 2\mathcal{H}_{f} + 2\mathcal{H}_{g} - \mathcal{E}_{1} - \mathcal{E}_{2} - \mathcal{E}_{3} - \mathcal{E}_{4} - 2\mathcal{E}_{5} - \mathcal{E}_{8},\\
\mathcal{E}_{7} \mapsto 2\mathcal{H}_{f} + 3\mathcal{H}_{g} - \mathcal{E}_{1} - \mathcal{E}_{2} - \mathcal{E}_{3} - \mathcal{E}_{4} - 2\mathcal{E}_{5} - \mathcal{E}_{6} - 2\mathcal{E}_{8},\\
\mathcal{E}_{8} \mapsto \mathcal{H}_{g} - \mathcal{E}_{5}. 
\end{gather*}
\end{Proposition}
\begin{proof}The decomposition of the mapping $\psi_{*}$ is obtained in the same way as in Proposition~\ref{prop:decomp-QRT}, and from there the action on $\operatorname{Pic}(\mathcal{X})$ immediately follows.
\end{proof}

\begin{Remark} Note that the action of $\psi_{*}$ on $\operatorname{Pic}(\mathcal{X})$ can be computed directly from the mapping~\eqref{eq:ex-ST}, however that computation is very complicated and has to be done using computer algebra. The approach of Proposition~\ref{prop:decomp-ST} is significantly simpler.
\end{Remark}

\subsection{Equivalence of the two dynamics}\label{sub:equivalence_of_the_two_dynamics}

We are now in the position to prove the main result of this paper.

\begin{Theorem}\label{thm:equivalence} The elementary Schlesinger transformation dynamics given by \eqref{eq:ex-ST} and the standard $($deautonomized QRT$)$ dynamics given by~\eqref{eq:ex-QRT} are \emph{equivalent} through
 the explicit change of variables transforming one equation into the other:
\begin{alignat*}{2}
 f&= \dfrac{x \big(y - \theta_{1}^{1}\big) - \big(\kappa_{1} + \theta_{0}^{2} + \theta_{1}^{1}\big) y}{y + \kappa_{1} + \theta_{0}^{2}},&\qquad \bar{f}&=\dfrac{\bar{x} \big(\bar{y} - \big(\theta_{1}^{1}+1\big)\big) -
 \big(\kappa_{1} + \theta_{0}^{2} + \theta_{1}^{1}+1\big) \bar{y}}{\bar{y} + \kappa_{1} + \theta_{0}^{2}},\\
 g&= \dfrac{x (y + \kappa_{1} + \theta_{0}^{1}) + \big(\theta_{0}^{1} - \theta_{0}^{2}\big)y }{x- \kappa_{1} - \theta_{0}^{2}},& \qquad
 \bar{g}&=\dfrac{\bar{x} \big(\bar{y} + \kappa_{1} + \theta_{0}^{1}-1\big) + \big(\theta_{0}^{1}-1 - \theta_{0}^{2}\big)\bar{y} }{\bar{x}- \kappa_{1} - \theta_{0}^{2}}.
 \end{alignat*}
Note that this change of variables also affects the identification of parameters, which becomes
\begin{alignat*}{4}
b_{1} &=-\kappa_{1} - \theta_{0}^{1} - \theta_{1}^{1},&\qquad
b_{2} &= \kappa_{2} + \theta_{0}^{2},&\qquad
b_{3} &= \kappa_{3} + \theta_{0}^{2},&\qquad b_{4} &= 0,\\
b_{5} &= \theta_{0}^{1} - \theta_{0}^{2},&\qquad
b_{6} &= \kappa_{1} + \theta_{0}^{1} + \theta_{1}^{2},&\qquad
b_{7} &=\theta_{1}^{1},&\qquad
b_{8} &= \kappa_{1} + \theta_{1}^{1}-1,
\end{alignat*}
and the parameter evolution $\bar{\theta}_{0}^{1} = \theta_{0}^{1} - 1$, $\bar{\theta}_{1}^{1} = \theta_{1}^{1} +1$ gives the standard evolution of the parame\-ters~$b_{i}$:
\begin{alignat*}{3}
\bar{b}_{1} &=-\kappa_{1} - \theta_{0}^{1} - \theta_{1}^{1} = b_{1},\!\!&\qquad
\bar{b}_{2} &= \kappa_{2} + \theta_{0}^{2} = b_{2},&\qquad
\bar{b}_{3} &= \kappa_{3} + \theta_{0}^{2} = b_{3}, \\
\bar{b}_{4} &= 0 = b_{4},&\qquad
\bar{b}_{5} &= \theta_{0}^{1} - \theta_{0}^{2}-1 = b_{5} - 1,\!\! &\qquad
\bar{b}_{6} &= \kappa_{1} + \theta_{0}^{1} + \theta_{1}^{2}-1 = b_{6} - 1,\\
\bar{b}_{7} &=\theta_{1}^{1}+1 = b_{7} + 1,&\qquad
\bar{b}_{8} &= \kappa_{1} + \theta_{1}^{1} = b_{8}+ 1.
\end{alignat*}
The inverse change of variables is given by the same formulas $($with the corresponding change of variables and parameters$)$.
\end{Theorem}
\begin{proof}
Comparing the decomposition \eqref{eq:decomp-QRT} and \eqref{eq:decomp-ST} of the two mappings,
\begin{gather*}\begin{split}&
\varphi = r\circ w_{5}\circ ({\color{purple} w_{2}\circ w_{6}\circ w_{5}\circ w_{3}\circ w_{2}\circ w_{4}\circ w_{3}\circ w_{1}\circ w_{2}\circ w_{5}\circ w_{0}\circ w_{1}\circ w_{2}\circ w_{6}} )\circ w_{5},\\
&\psi = r\circ w_{1}\circ ({\color{purple} w_{2}\circ w_{6}\circ w_{5}\circ w_{3}\circ w_{2}\circ w_{4}\circ w_{3}\circ w_{1}\circ w_{2}\circ w_{5}\circ w_{0}\circ w_{1}\circ w_{2}\circ w_{6}} )\circ w_{3},
\end{split}
\end{gather*}
we immediately see that $\varphi = r\circ w_{5} \circ w_{1} \circ r^{2} \circ \psi \circ w_{3} \circ w_{5} = (w_{5}\circ w_{3})\circ \psi\circ (w_{5}\circ w_{3})^{-1}$ (note that in our case $w_{3} \circ w_{5} = w_{5}\circ w_{3}$). Then the mapping $w_{5}\circ w_{3}$, that can be easily computed from Theorem~\ref{thm:bir-weyl}, gives us the required change of variables. The resulting equivalence can then be verified by direct computation.
\end{proof}

\begin{Remark}\label{rem:conj}Note that it is also possible to establish the equivalence of two dynamics in the following way. For a root $\alpha\in W = W(\mathcal{D}_{2})$ we can define the so-called Kac's translation $\mathbf{t}_{\alpha}\in W$ that acts on the symmetry roots as $\mathbf{t}_{\alpha}\colon \beta \mapsto \beta+(\alpha,\beta)\delta$. Then, as shown in Section~6.5 of \cite{Kac:1985:ILA}, for $w\in W$, $\mathbf{t}_{w(\alpha)}=w^{-1}\circ \mathbf{t}_{\alpha}\circ w$:
\begin{gather*}
\big(w^{-1}\circ \mathbf{t}_{\alpha}\circ w\big) (\beta)=w^{-1}(w(\beta)+(\alpha,w(\beta))\delta)=\beta+(\alpha,w(\beta))\delta=\mathbf{t}_{w(\alpha)}(\beta).
\end{gather*}
Thus, the \emph{norm} $|\mathbf{t}_{\alpha}|^{2}:=-(\alpha\bullet\alpha)$ (where the negative sign reflects our sign choice for the Cartan matrix) is \emph{invariant} under conjugations.

This property can be extended to $\widetilde{W} = \widetilde{W}(\mathcal{D}_{2})$ as follows. Let $Q_{\mathbb{Q}} = Q\otimes{\mathbb{Q}}$ be the $\mathbb{Q}$-vector space spanned by the symmetry roots and let $\mathbf{t}_{\alpha}\in \widetilde{W}$, $\alpha\in Q_{\mathbb{Q}}$, act on $Q_{\mathbb{Q}}$ in the same way, $\mathbf{t}_{\alpha}\colon \beta \mapsto \beta+(\alpha\bullet\beta)\delta$. Then, as before, $\mathbf{t}_{w(\alpha)}=w^{-1}\circ \mathbf{t}_{\alpha}\circ w$ for any $w\in \widetilde{W}$ and so the norm is again preserved under conjugation. Thus, if $|\mathbf{t}_{\alpha}|\neq |\mathbf{t}_{\beta}|$ for $\alpha, \beta\in Q_{\mathbb{Q}}$, then $\mathbf{t}_{\alpha}$ and $\mathbf{t}_{\beta}$ are not conjugate, otherwise, if we can find $w\in \widetilde{W}$ such that $\beta=w(\alpha)$, then $\mathbf{t}_{\beta}=w^{-1}\circ \mathbf{t}_{\alpha}\circ w$.

In our case, it is easy to see that
 \begin{gather*}
\varphi_{*} = \mathbf{t}_{\frac{1}{3}(2\alpha_{5} + \alpha_{6} - 2\alpha_{3}-\alpha_{4})} \colon \
 (\alpha_0,\alpha_1,\alpha_2,\alpha_3,\alpha_4,\alpha_5,\alpha_6) \mapsto (\alpha_0,\alpha_1,\alpha_2,\alpha_3 + \delta,\alpha_4,\alpha_5 - \delta,\alpha_6),\\
\psi_{*}= \mathbf{t}_{\frac{1}{3} (\alpha_{3} - \alpha_{4} - \alpha_{5} + \alpha_{6} )} \colon \\
\qquad{} (\alpha_0,\alpha_1,\alpha_2,\alpha_3,\alpha_4,\alpha_5,\alpha_6) \mapsto
 (\alpha_0,\alpha_1,\alpha_2,\alpha_3 - \delta,\alpha_4 + \delta, \alpha_5 + \delta,\alpha_6 - \delta),
 \end{gather*}
and
\begin{gather*}
\big|\mathbf{t}_{\frac{1}{3} (2\alpha_{5} + \alpha_{6} - 2\alpha_{3}-\alpha_{4} )}\big| =\big|\mathbf{t}_{\frac{1}{3}(\alpha_{3} - \alpha_{4} - \alpha_{5} + \alpha_{6} )}\big| = \frac{4}{3}.
 \end{gather*}
Next, note that
 \begin{gather*}
 (w_{3}\circ w_{5})(\alpha_{3} - \alpha_{4} - \alpha_{5} + \alpha_{6}) = 2 \alpha_{5} + \alpha_{6} - 2 \alpha_{3} - \alpha_{4}.
 \end{gather*}
Hence $\psi=(w_{3}\circ w_{5})\circ \varphi \circ (w_{3}\circ w_{5})^{-1}$, exactly as we obtained previously.
\end{Remark}

\section{Conclusion}\label{sec:conclusion}
In this paper we showed how to determine whether two different discrete Painlev\'e dynamics are equivalent and if so, how to transform one into the other. The key technique is to use the algebraic structures underlying the theory of discrete Painlev\'e equations, especially the birational representation of the extended affine Weyl group of symmetries.

\subsection*{Acknowledgements}

A.D.'s work was partly supported by the University of Northern Colorado 2015 Summer Support Initiative. T.T.\ was supported by the Japan Society for the Promotion of Science, Grand-in-Aid~(C) (17K05271). We thank N.~Nakazono for explaining to us the techniques discussed in Remark~\ref{rem:conj}. We are very grateful to A.~Ramani, R.~Willox, and the referees for useful suggestions and corrections.

\pdfbookmark[1]{References}{ref}
\LastPageEnding


\begin{thebibliography}{99}
\footnotesize\itemsep=0pt

\bibitem{Boa:2009:QADPE}
Boalch P., Quivers and difference {P}ainlev\'e equations, in Groups and
 Symmetries, \textit{CRM Proc. Lecture Notes}, Vol.~47, Amer. Math. Soc.,
 Providence, RI, 2009, 25--51, \href{https://arxiv.org/abs/0706.2634}{arXiv:0706.2634}.

\bibitem{Bor:2003:DPDPE}
Borodin A., Discrete gap probabilities and discrete {P}ainlev\'e equations,
 \href{https://doi.org/10.1215/S0012-7094-03-11734-2}{\textit{Duke Math.~J.}} \textbf{117} (2003), 489--542,
 \href{https://arxiv.org/abs/math-ph/0111008}{math-ph/0111008}.

\bibitem{BorBoy:2003:DOTFPIDOPE}
Borodin A., Boyarchenko D., Distribution of the first particle in discrete
 orthogonal polynomial ensembles, \href{https://doi.org/10.1007/s00220-002-0767-3}{\textit{Comm. Math. Phys.}} \textbf{234}
 (2003), 287--338, \href{https://arxiv.org/abs/math-ph/0204001}{math-ph/0204001}.

\bibitem{CarTak:2013:ANOMORSOFBDS}
Carstea A.S., Takenawa T., A note on minimization of rational surfaces obtained
 from birational dynamical systems, \href{https://doi.org/10.1080/14029251.2013.862432}{\textit{J.~Nonlinear Math. Phys.}}
 \textbf{20} (2013), suppl.~1, 17--33, \href{https://arxiv.org/abs/1211.5393}{arXiv:1211.5393}.

\bibitem{Dol:1983:WGACT}
Dolgachev I.V., Weyl groups and {C}remona transformations, in Singularities,
 {P}art~1 ({A}rcata, {C}alif., 1981), \textit{Proc. Sympos. Pure Math.},
 Vol.~40, Amer. Math. Soc., Providence, RI, 1983, 283--294.

\bibitem{DzhSakTak:2013:DSTTHFADPE}
Dzhamay A., Sakai H., Takenawa T., Discrete {S}chlesinger transformations,
 their {H}amiltonian formulation, and difference {P}ainlev\'e equations,
 \href{https://arxiv.org/abs/1302.2972}{arXiv:1302.2972}.

\bibitem{DzhTak:2015:GAORFSTTDPE}
Dzhamay A., Takenawa T., Geometric analysis of reductions from {S}chlesinger
 transformations to difference {P}ainlev\'e equations, in Algebraic and
 Analytic Aspects of Integrable Systems and {P}ainlev\'e Equations,
 \href{https://doi.org/10.1090/conm/651/13044}{\textit{Contemp. Math.}}, Vol.~651, Amer. Math. Soc., Providence, RI, 2015,
 87--124, \href{https://arxiv.org/abs/1408.3778}{arXiv:1408.3778}.

\bibitem{GraRam:2004:DPER}
Grammaticos B., Ramani A., Discrete {P}ainlev\'e equations: a review, in
 Discrete Integrable Systems, \href{https://doi.org/10.1007/978-3-540-40357-9_7}{\textit{Lecture Notes in Phys.}}, Vol.~644,
 Springer, Berlin, 2004, 245--321.

\bibitem{GraRamOht:2003:AUDOTAQVADIEATST}
Grammaticos B., Ramani A., Ohta Y., A unified description of the asymmetric
 {$q\text{-P}_{\rm V}$} and {$d\text{-P}_{\rm IV}$} equations and their
 {S}chlesinger transformations, \href{https://doi.org/10.2991/jnmp.2003.10.2.5}{\textit{J.~Nonlinear Math. Phys.}} \textbf{10}
 (2003), 215--228, \href{https://arxiv.org/abs/nlin.SI/0310050}{nlin.SI/0310050}.

\bibitem{GraRamPap:1991:IMHPP}
Grammaticos B., Ramani A., Papageorgiou V., Do integrable mappings have the
 {P}ainlev\'e property?, \href{https://doi.org/10.1103/PhysRevLett.67.1825}{\textit{Phys. Rev. Lett.}} \textbf{67} (1991),
 1825--1828.

\bibitem{IwaKimShiYos:1991:FGP}
Iwasaki K., Kimura H., Shimomura S., Yoshida M., From {G}auss to {P}ainlev\'e:
 a~modern theory of special functions, \href{https://doi.org/10.1007/978-3-322-90163-7}{\textit{Aspects of Mathematics}}, Vol.~E16, Friedr. Vieweg \& Sohn, Braunschweig, 1991.

\bibitem{Kac:1985:ILA}
Kac V.G., Infinite-dimensional {L}ie algebras, 3rd~ed., \href{http://dx.doi.org/10.1017/CBO9780511626234}{Cambridge University
 Press}, Cambridge, 1990.

\bibitem{KajNouYam:2017:GAOPE}
Kajiwara K., Noumi M., Yamada Y., Geometric aspects of {P}ainlev\'e equations,
 \href{https://doi.org/10.1088/1751-8121/50/7/073001}{\textit{J.~Phys.~A: Math. Theor.}} \textbf{50} (2017), 073001, 164~pages,
 \href{https://arxiv.org/abs/1509.08186}{arXiv:1509.08186}.

\bibitem{Mas:2017:SOSOICFNMOTP}
Mase T., Studies on spaces of initial conditions for nonautonomous mappings of
 the plane, \href{https://arxiv.org/abs/1702.05884}{arXiv:1702.05884}.

\bibitem{Nou:2004:PETS}
Noumi M., Painlev\'e equations through symmetry, \textit{Translations of
 Mathematical Monographs}, Vol.~223, Amer. Math. Soc., Providence, RI, 2004.

\bibitem{NouYam:1998:AWGDDSAPE}
Noumi M., Yamada Y., Affine {W}eyl groups, discrete dynamical systems and
 {P}ainlev\'e equations, \href{https://doi.org/10.1007/s002200050502}{\textit{Comm. Math. Phys.}} \textbf{199} (1998),
 281--295, \href{https://arxiv.org/abs/math.QA/9804132}{math.QA/9804132}.

\bibitem{Oka:1979:SLFAAEDSOAPCFDPP}
Okamoto K., Sur les feuilletages associ\'es aux \'equations du second ordre \`a
 points critiques fixes de {P}.~{P}ainlev\'e, \href{https://doi.org/10.4099/math1924.5.1}{\textit{Japan.~J. Math. (N.S.)}}
 \textbf{5} (1979), 1--79.

\bibitem{Pai:1902:SLEDDSOEDSDLGEU}
Painlev\'e P., Sur les \'equations diff\'erentielles du second ordre et d'ordre
 sup\'erieur dont l'int\'egrale g\'en\'erale est uniforme, \href{https://doi.org/10.1007/BF02419020}{\textit{Acta Math.}}
 \textbf{25} (1902), 1--85.

\bibitem{QuiRobTho:1988:IMASE}
Quispel G.R.W., Roberts J.A.G., Thompson C.J., Integrable mappings and soliton
 equations, \href{https://doi.org/10.1016/0375-9601(88)90803-1}{\textit{Phys. Lett.~A}} \textbf{126} (1988), 419--421.

\bibitem{Rai:2013:GHSORS}
Rains E.M., Generalized {H}itchin systems on rational surfaces,
 \href{https://arxiv.org/abs/1307.4033}{arXiv:1307.4033}.

\bibitem{RamGraHie:1991:DVOTPE}
Ramani A., Grammaticos B., Hietarinta J., Discrete versions of the {P}ainlev\'e
 equations, \href{https://doi.org/10.1103/PhysRevLett.67.1829}{\textit{Phys. Rev. Lett.}} \textbf{67} (1991), 1829--1832.

\bibitem{Sak:2001:RSAWARSGPE}
Sakai H., Rational surfaces associated with affine root systems and geometry of
 the {P}ainlev\'e equations, \href{https://doi.org/10.1007/s002200100446}{\textit{Comm. Math. Phys.}} \textbf{220} (2001),
 165--229.

\bibitem{TamRamGra:2003:DI}
Tamizhmani K.M., Ramani A., Grammaticos B., Tamizhmani T., Discrete
 integrability, in Classical and Quantum Nonlinear Integrable Systems: Theory
 and Application, Editor A.~Kundu, IOP Publishing, Briston, 2003, 64--94.

\bibitem{TraWid:1994:LDATAK}
Tracy C.A., Widom H., Level-spacing distributions and the {A}iry kernel,
 \href{https://doi.org/10.1007/BF02100489}{\textit{Comm. Math. Phys.}} \textbf{159} (1994), 151--174,
 \href{https://arxiv.org/abs/hep-th/9211141}{hep-th/9211141}.

\end{thebibliography}
\end{document}